\documentclass[letterpaper,12pt,reqno]{amsart}

\usepackage{amsmath}

\usepackage{geometry}
\usepackage{amsfonts,amssymb,amsthm,enumitem}
\usepackage{mathtools}
\usepackage{nicefrac,setspace}
\usepackage{xcolor}
\usepackage{tikz}
\usepackage[framemethod=tikz]{mdframed}
\usepackage{pgfplots}
\usepackage{subfig}
\usepackage{graphicx,float}
\pgfplotsset{compat=1.18}
\usepackage[font=footnotesize,labelfont=bf]{caption}
\usepackage{hyperref}       
\usepackage{url}            
\usepackage{booktabs}       
\usepackage{amsfonts}       
\usepackage{nicefrac}       
\usepackage{microtype}      

\newtheorem{theorem}{Theorem}

\newtheorem*{theorem*}{Theorem}
\newtheorem*{claim*}{Claim}
\newtheorem{remark}{Remark}
\newtheorem*{lemma*}{Lemma}
\newtheorem{definition}{Definition}
\newtheorem{lemma}[theorem]{Lemma}

\newtheorem*{corollary*}{Corollary}

\newtheorem{proposition}[theorem]{Proposition}
\newtheorem{conjecture}[theorem]{Conjecture}

\newcommand{\no}{\Pi_{\text{NO}}}
\newcommand{\yes}{\Pi_{\text{YES}}}
\newcommand{\Cov}{\mathrm{Cov}}

\newcommand{\Pp}{{\mathcal{P}}}

\newcommand{\M}{{\mathcal{M}}}
\newcommand{\calC}{{\mathcal{C}}}

\newcommand{\N}{\mathbb{N}}
\newcommand{\eps}{\varepsilon}

\newcommand{\R}{\mathbb{R}}

\newcommand{\E}{\mathop \mathbb{E}}

\newcommand{\la}{\lambda}

\newcommand{\al}{\alpha}
\newcommand{\be}{\beta}
\newcommand{\ep}{\varepsilon}

\newcommand{\del}{\delta}

\renewcommand{\sp}{\mathbb{S}}
\renewcommand{\r}{\mathbb{R}}

\newcommand{\K}{\mathcal{K}}

\newcommand{\BB}{\mathfrak{B}}

\newcommand{\cp}{\mathbb{C}}

\title{Geometric Covering using Random Fields}

\author{Felipe Gon\c{c}alves}
\address{Department of Mathematics, IMPA}
\email{goncalves@impa.br}

\author{Daniel Keren}
\address{Department of Computer Science, University of Haifa}
\email{dkeren@cs.haifa.ac.il}

\author{Amit Shahar}
\address{Department of Computer Science, University of Haifa}
\email{ashaha16@campus.haifa.ac.il}

\author{Gal Yehuda}
\address{Department of Computer Science, Technion-IIT}
\email{ygal@technion.ac.il}

\begin{document}

\begin{abstract}
A set of vectors $S \subseteq \R^d$ is $(k_1,\varepsilon)$-clusterable if there are $k_1$ balls of radius $\varepsilon$ that cover $S$. 
A set of vectors $S \subseteq \R^d$ is $(k_2,\delta)$-far from being clusterable if there are at least $k_2$ vectors in $S$, with all pairwise distances at least $\delta$. 
We propose a probabilistic algorithm to distinguish between these two cases.
Our algorithm reaches a decision by only looking at the extreme values of a \emph{scalar valued} hash function, defined by a \emph{random field}, on $S$; hence, it is especially suitable in distributed and online settings.
An important feature of our method is that the algorithm is oblivious to the number of vectors: in the online setting, for example, the algorithm stores only a constant number of scalars, which is independent of the stream length.

We introduce random field hash functions, which are a key ingredient in our paradigm. 
Random field hash functions generalize \emph{locality-sensitive hashing} (LSH). 
In addition to the LSH requirement that ``nearby vectors are hashed to similar values", our hash function also guarantees that the ``hash values are (nearly) independent random variables for distant vectors". 
We formulate necessary conditions for the kernels which define the random fields applied to our problem, as well as a measure of kernel optimality, for which we provide a bound.
Then, we propose a method to construct kernels which approximate the optimal one.
\end{abstract}

\maketitle

\section{Introduction} \label{section:introduction}
The question whether a set of vectors $S$ in Euclidean space is ``clusterable'' is of considerable interest in numerous applications. 
It is also known to be hard, in its various formulations.
For example, the \emph{geometric set cover problem} is $\textsc{NP}$-complete even in dimension two.

We consider the following promise-problem version: given a set $S \subseteq \R^d$, can $S$ be covered by $k_1$ balls of radios $\varepsilon$, or are there at least $k_2$ vectors in $S$ with mutual distances at least $\delta$?
The two cases are shown pictorially in Figures \ref{fig:condition_c},\ref{fig:condition_m}.

We motivate our approach by presenting two basic examples.
\subsection{Generalized Distinct Count}
Consider the following version of the classical Flajolet-Martin (FM hereafter) algorithm \cite{fl1985}, which allows to efficiently, and in constant space, count distinct items in a data stream. 
A uniform hash function from the object space $W$ to the unit interval $h:W \rightarrow I \triangleq [0,1]$ is defined, and the number of distinct items is estimated by  $\frac{1}{\min(h(x))_{x \in S}}$. 
Standard ``stabilization'' techniques (e.g. median of averages) can be applied to render this estimate more accurate and robust. 
As the estimate uses only the minimal value of the hash function, it requires only a constant-size communication overhead to implement in a distributed setting, as the nodes can agree on $h$ in advance, and can only broadcast a fixed-size table required to obtain a stable estimate of the minimum \cite{DBLP:conf/stoc/AlonMS96}.

Assume now that $W$ is not composed of discrete objects, but scalars (real numbers or vectors with scalar components). 
In this case, typically, objects whose corresponding scalars are very close to each other, can be considered identical, or nearly identical, and it makes little sense to count them as distinct elements.
Instead, we would like to view close vectors as a single element, and count the tight clusters they form. 
In this context, clustering can be seen as a generalization of the Distinct Count problem (``how many distinct elements are in a given set?''), when generalized to a continuous space, with ``distinct'' replaced by ``of distance greater than some threshold''.

\subsection{Inferring on the Size of a Distributed Set}
Assume a large set of data vectors are distributed among nodes, and these nodes wish to determine, with high probability, whether the union of their data sets can be covered by a ball of radius $\varepsilon$, or whether an ``opposite" condition holds -- for example, that the union contains at least two vectors whose distance from each other is at least $(2+c)\varepsilon$, for some $c>0$.
Further, the decision is to be taken with minimal communication overhead.

Our proposed algorithm will achieve the desired goal with communication overhead that is independent of the overall size of the data, and is constant for each node, regardless of the number of data vectors.
Actually, we will answer a more general question -- whether the vectors can be covered by a certain number of balls, or the ``opposite" condition, being the presence of a set of vectors with  pairwise distance greater than a given threshold.

\begin{figure}
    \centering
    \begin{minipage}{0.45\textwidth}
        \centering
        \includegraphics[scale=0.6]{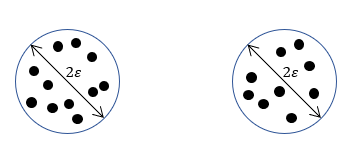}
        \vspace{0.65cm}
        \caption{ An example of clusterable data. That is, data that can be covered with 2 balls with radius $\varepsilon$.}
    \label{fig:condition_c}
    \end{minipage}
    \begin{minipage}{0.45\textwidth}
        \centering
        \includegraphics[scale=0.5]{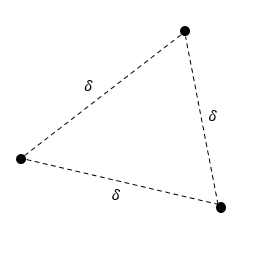}
        \vspace{-0.45cm}
        \caption{An example of unclusterable data: three vectors with pairwise distance $\delta>2\varepsilon$.}
    \label{fig:condition_m}
    \end{minipage}
\end{figure}

\subsection{Our Contribution} 
We propose a probabilistic approach for the clustering promise problem. 
Our main contribution is bringing together ideas from the theory of random fields and streaming algorithms. 
More precisely, we show how to use random fields in order to construct a hash function $h : \R^d \rightarrow \R$ with the following properties. 

\begin{enumerate}
    \item For every $x \in \R^d$, $h(x)$ is a random variable, with identical distribution.
    \item Roughly speaking, if $\|x-y\| \leq \varepsilon$, then $|h(x)-h(y)|$ is ``small''. 
    This will ensure that ``tight'' clusters will be mapped to nearby values. 
    This property also plays a major role in the \emph{locality sensitive hashing} (LSH) scheme \cite{DBLP:conf/stoc/IndykM98}. \label{hash_function:cond1}
    \item If $\|x-y\| \geq \delta$ , $h(x)$ and $h(y)$ should be ``nearly independent''. 
    This property ensures that for $x,y$ which are far apart, $h$ behaves like the hashing used in FM, i.e. it will treat $x,y$ as distinct objects. 
    This property is not required in LSH, and indeed does not hold for the hash functions used in its implementation (Section \ref{section:efficient_random_fields}).
    \label{hash_function:cond2}
\end{enumerate}

Intuitively, a hash function satisfying the above conditions can answer the clustering promise problem: for distant vectors $x,y$, the random variables $h(x)$ and $h(y)$ are ``almost'' independent. 
In this case, similarly to FM hashing, some simple property of the distribution of $h(x)$ (e.g. its maximum over some set) carries enough information about the structure of the set (e.g. its cardinality).
On the other hand, for close vectors, the random variables $h(x)$ and $h(y)$ are highly correlated.
This property allows to treat close vectors as if they are identical.

We show that the function $h$ can be defined using a random field. 
While we defer the formal definition for later, a random field over $\R^d$ can be thought of as a function that maps each $x \in \R^d$ to random variable $f(x)$ over $\R$. 
We show explicitly how to construct such random fields.

There is an inherent tension between Condition ~\ref{hash_function:cond1} and ~\ref{hash_function:cond2}: while Condition ~\ref{hash_function:cond1} requires $h(x)$ and $h(y)$ to have ``large'' covariance if $\|x-y\| \leq \varepsilon$, Condition ~\ref{hash_function:cond2} requires $h(x)$ and $h(y)$ to have ``small'' covariance if $\|x-y\| > \delta$. 
Ideally, the covariance function of $h(x)$ and $h(y)$, as a function of the distance, should ``drop as fast as possible" between the distance values $\varepsilon$ and $\delta$. 
However, this cannot be achieved by the trivial solution of taking a step function which equals 1 at $\varepsilon$ and 0 at $\delta$, as it will not define a distribution, since the corresponding covariance matrix will not be positive semidefinite; more on this later (Section ~\ref{section:lower_bound_optimal_kernels}).

We analyze the rate at which the covariance function can decrease, prove lower bounds, and show how to attain a solution which is near optimal.

\subsection{Previous Work}
A work by Alon et al. \cite{testingofclustering} investigates the clustering problem in the context of property testing. 
There are similarities in the problem definition, in particular, the goal is to answer whether a set of vector is clusterable, or ``far'' from being clusterable.
However, there is a key difference between the property testing setting and the promise problem we consider.
The objective in the property testing setting is to answer the if an ``unknown'' set $S \subseteq \R^d$ is clusterable or far from being clusterable, while querying as few vectors from $S$ as possible.
More generally, the goal of a property testing algorithm is testing whether a set $S$ satisfies a property, or is $\varepsilon$-far from satisfying a property, and the objective is to make less than $|S|$ queries to $S$ -- ideally, the number of queries is independent of $|S|$ and depends only on $\varepsilon$. 
The running time, as a function of $\varepsilon$, is usually less significant.
Our goal in this work, however, is to devise algorithms for answering the clustering promise problem quickly, while having full access to the dataset. 

Additional work in the context of property testing is given in \cite{helly_theorems_propertytesting}.
The presented approach is also applicable to other symmetric shapes (not just balls), but it only addresses the problem of covering with a \emph{single} shape. 
We have not been able to find work along the major direction of this paper, i.e. the application of random fields to the clusterability problem.

For the simpler problem of \emph{distinct count} (which is equivalent to covering with balls of zero radii), a well-known approach was presented in FM \cite{fl1985}. 
One version of FM is to randomly hash the dataset $S$ into the interval $[0,1]$.
A standard exercise proves that if there are $m$ distinct elements, the expectation of the minimum value of the hash is $\frac{1}{m+1}$.
While our approach has a similar flavor -- we look at the maximal value of a hash function -- it is different in a crucial aspect: FM hashing is oblivious to the distance between vectors, i.e. even when provided with a very tight cluster, its output will still depend only on the number of vectors in the cluster; this is due to the fact that the values of the hash function at any distinct vector are independent random variables.


The property of mapping nearby vectors to similar values might remind one of \emph{locality sensitive hashing} (LSH) method \cite{DBLP:conf/stoc/IndykM98}, especially the version presented in \cite{DBLP:conf/compgeom/DatarIIM04}, which deals with real-valued input. 
However, while the requirement ``nearby vectors are hashed to similar values" is indeed satisfied in LSH, the requirement ``the hash values are (nearly) independent for far apart vectors" is not; hence, as will be elaborated later, LSH is inappropriate for our problem. 
Still, there appears to be an interesting relation between LSH and one of the hash functions we propose to use (Section ~\ref{section:efficient_random_fields}).

\subsection{Paper Structure}
After introducing some basic notions in Section ~\ref{section:perlim}, we present,
in Section ~\ref{section:hash_function_clustring}, our main contribution and
algorithm. 
Section ~\ref{section:efficient_random_fields} discusses a new class of random fields which can also be applied to our problem and that are easier to construct. 
Sections ~\ref{section:lower_bound_optimal_kernels}, ~\ref{section:higher_dimenstion}, ~\ref{section:numerical_eval} are technical in nature, and concern
the construction of optimal kernels, in one and many dimensions. 
Conclusions are offered in Section ~\ref{section:conclusions}.

\section{Preliminaries} \label{section:perlim}
In this section we fix notations, formally define our problem, and give standard definitions from the theory of random fields.   

\subsection{Notations}
All vector norms are $\ell^2$-norms, unless otherwise stated. 
The inner product of vectors $x,y$ is denoted $x \cdot y$.
For a vector $v \in \R^d$, and $\varepsilon \in \R_{+}$, we denote by $B(v, \varepsilon)$ a ball of radius $\varepsilon$ and center $v$. 
In places where the center is irrelevant, we omit the center from the notation.

\subsection{Clustering promise problem}
We define the clustering promise problem $\Pi = \Pi_{\text{YES}} \cup \Pi_{\text{NO}}$. 
Let $S \subseteq \R^d$, $k_1 \in \N$ and $\varepsilon \in \R_{+}$. 

\begin{definition}
    Given a set $S \subseteq \R^d$ and parameters $k_1 \in \N$, $\varepsilon \in \R_{+}$, we say that $S$ is \emph{$(k_1,\varepsilon)$-clusterable} if there exist $k_1$ balls $\mathcal{B} = \{B_1, \ldots, B_{k_1}\}$ of radius $\varepsilon$, such that for all $x \in S$ there exists a ball $B \in \mathcal{B}$ such that $x \in B$. 
\end{definition}

We denote by $\Pi_{\text{YES}} = \{S \subseteq \R^d : S \text{ is } (k_1, \varepsilon) \text{ clusterable}\}$. 
For clarity, we omit the parameters $\varepsilon, k_1$ from the set notation.

\begin{definition}
    Given a set $S \subseteq \R^d$ and parameters $k_2 \in \N$, $\delta \in \R_{+}$, we say that $S$ is \emph{$(k_2,\delta)$-far from being clusterable} if there are at least $k_2$ vectors in $S$, with all pairwise distances at least $\delta$.
\end{definition}
We denote by $\Pi_{\text{NO}} = \{S \subseteq \R^d : S \text{ is } (k_2, \delta) \text{-far from being clusterable}\}$.

\subsection{The Relation Between the Conditions \texorpdfstring{$\Pi_{\text{YES}}$}{Py} and \texorpdfstring{$\Pi_{\text{NO}}$}{Pn}}
The complementary nature of the conditions is realized by the following Lemma.

\begin{lemma} \label{lemma:relation}
If $S \subset \R^d$ has at most $\ell$ vectors with pairwise distance at least $\delta$, then its minimal cover with balls of radius $\varepsilon$ is of size at most  $\ell c(d)\left(\frac{\delta}{\varepsilon}\right)^d$, where $c(d)$ depends only on $d$.
\end{lemma}

\begin{proof}
Suppose there are $\ell$ vectors with mutual distances at least $\delta$ in $S$, but not $\ell+1$ such vectors. 
Then, $S$ can be covered by balls of radius $\delta$ around the $\ell$ vectors (since if these balls do not cover $S$, there will be a vector at distance at least $\delta$ from all the $\ell$ vectors, a contradiction).

Next, we cover these $\delta$-balls with $\varepsilon$-balls (Fig. \ref{fig:condition_c_proof}). The number of $\varepsilon$-balls required is equal to $c(d)\left(\frac{\delta}{\varepsilon}\right)^d$, where $d$ is the dimension of the ambient Euclidean space and $c(d)$ a function which depends only on $d$ \cite{DBLP:journals/dcg/Verger-Gaugry05}. 
Multiplying by $\ell$ completes the proof.
\end{proof}

\begin{figure}
  \begin{minipage}[c]{0.5\textwidth}
  \centering
    \includegraphics[scale=0.5]{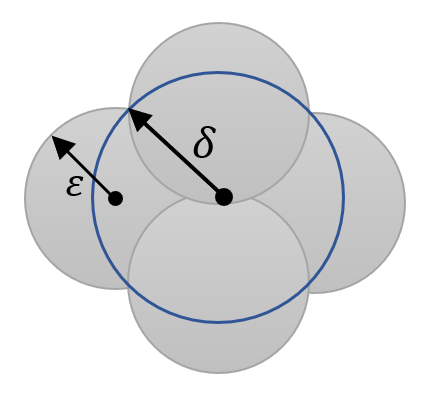}
  \end{minipage}\hfill
  \begin{minipage}[c]{0.5\textwidth}
  \centering
    \caption{A sketch of the proof of Lemma \ref{lemma:relation}. First, $S$ is covered with balls of radius $\delta$; then, these balls are covered with balls of radius $\eps$.}
     \label{fig:condition_c_proof}
  \end{minipage}
\end{figure}

\subsection{Random Fields}
We recall a few definitions relating to the notion of a random field.

\begin{definition}
(Random Field, \cite{adler2009random,randomfields_physics,brain_fields}). 
Given a probability space $(\Omega, \mathcal{F}, P)$,  a random field $T(x)$ defined on $\R^d$ is a function such that for every $x \in \R^d$, $T(x)$ is a random variable on $(\Omega, \mathcal{F}, P)$.
\end{definition}

\begin{definition}
(Covariance of a random field, \cite{adler2009random}). 
The covariance function $\Cov\left(x,y\right)$ of a random field is defined as 
$$\Cov(x,y) = \E[(T(x) - \E[T(x)])(T(y) - \E[T(y)])]$$
\end{definition}

A random field is {\em isotropic and stationary} if there exists a function $f$ such that $\Cov(x,y) = f(\|x-y\|)$ \footnote{In fact, it is possible to define these two properties separately, but in this work we are only interested in fields which are both isotropic and stationary.}. 
Since we are interested in properties which depend only on the distances between the vectors of $S$, we will only consider such fields.

\section{Hash Functions for Geometric Clustering} \label{section:hash_function_clustring}
At the foundation of our algorithm is a random field which, loosely speaking, maps nearby vectors to similar values, while the images of far away vectors are mapped to independent values. 
We note here that the ability of a random field to map nearby vectors to nearby values is restricted by the following result \cite{Lin2011}:
\begin{theorem} Let $f(t)$ be the characteristic function of a non-degenerate
distribution. 
Then there are positive constants $\alpha,\beta$ such that $f(t) \leq 1-\alpha t^2$ for $|t| \leq \beta$.
\end{theorem}
Intuitively speaking, this result limits the capability of a random field to
``compress'' the images of nearby vectors.

 A natural candidate for a random field is the \emph{Gaussian
  Random Field}(GRF), a key construct in probability, applied mathematics, and physics
\cite{adler2009random, brain_fields}. 
The Gaussian random field (Definition \ref{def:grf}) is also obtaining -- up to a constant -- the best ``compressibility'' possible. 



\subsection{Gaussian Random Fields}
\begin{definition} \label{def:grf}
  A (stationary) Gaussian Random Field on a set $D \subset \R^d$ is a class of random
  real-valued functions on $D$, satisfying the following:\\For every
  $x_1 \ldots x_n, x_i \in D$, the vector $(f(x_1)... f(x_n))$ obeys a normal
  distribution with zero mean and a covariance matrix $C$ given by $C_{i,j} = K(x_i,x_j)$,
  for some positive definite kernel $K(,)$.
\end{definition} 
  
Recall that a kernel is positive definite if it satisfies $\int_{D} K(x,y)f(x)f(y)dxdy \geq 0$ for every $f \in L_2(D)$; this also means that for every set of vectors $x_1 \ldots x_n$ in $D$, the matrix $K(x_i,x_j)$ is positive definite. 
Since our only interest is in \emph{stationary and isotropic} kernels -- i.e., which
depend only on $\|x-y\|$ -- we make the standard assumption that $K(x,y) \triangleq K\left(\|x-y\| \right)$.

The simplest -- yet very important -- GRF, is the one defined by
\[ K(t) = \exp\left(-\lambda\|t\|^2\right) \text{ , }\]
for some constant $\lambda > 0$. 
In this case, for each $x \in D$, $f(x)$ is a standard normal random variable, and the covariance matrix of $(f(x_1)... f(x_n))$ is given by $C_{i,j} = \exp\left(-\lambda \|x_i - x_j\|^2\right)$.
Note that there are two Gaussian functions involved: the first defines the covariance matrix of the second normal distribution. 
We shall refer to this field as the \emph{basic GRF}. 
In Sections ~\ref{section:lower_bound_optimal_kernels} and ~\ref{section:higher_dimenstion}, more general GRF's will be studied.

As demonstrated in Figure ~\ref{fig:GRF}, for both one and two dimensional fields, as $\lambda$ increases, the field become more ``rough'', and the correlation between its values at $x,y$ drops faster as $\|x-y\|$ increases.

\begin{figure}[ht!]
    \centering  
    \includegraphics[scale=0.15]{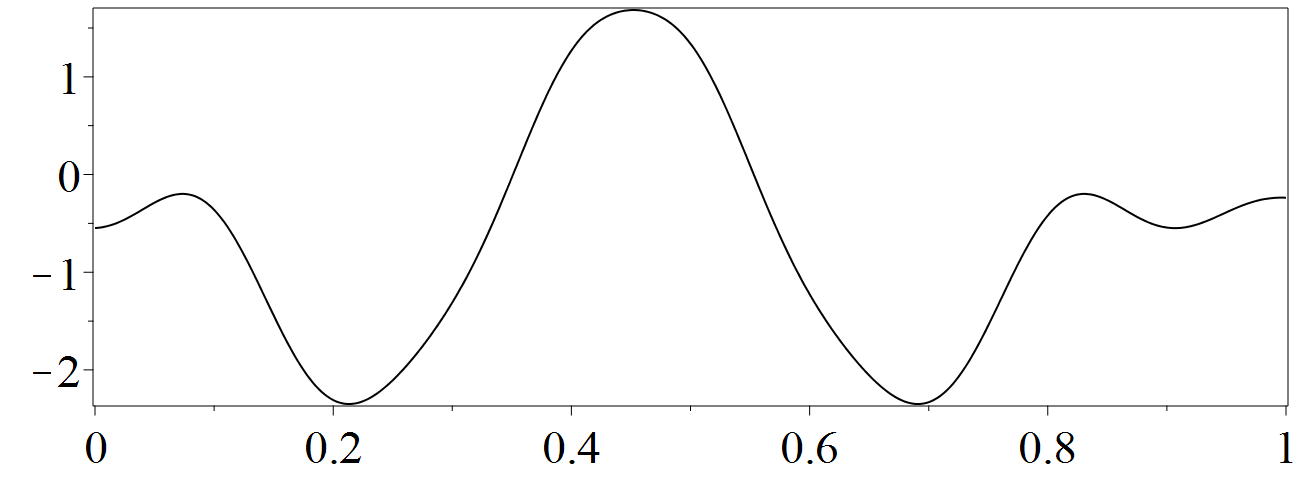}   \includegraphics[scale=0.15]{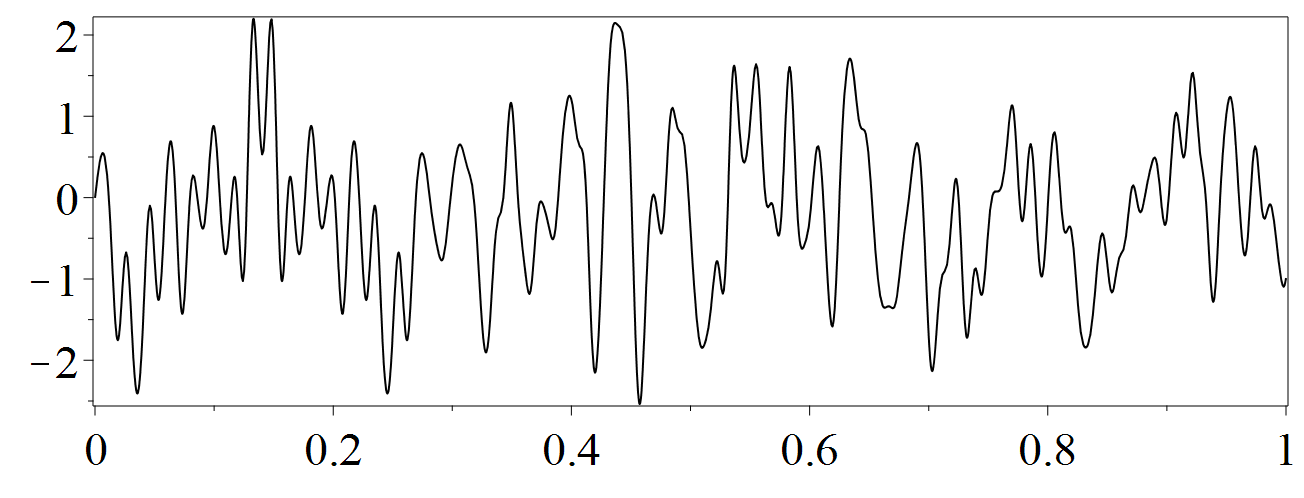} 
    \includegraphics[scale=0.5]{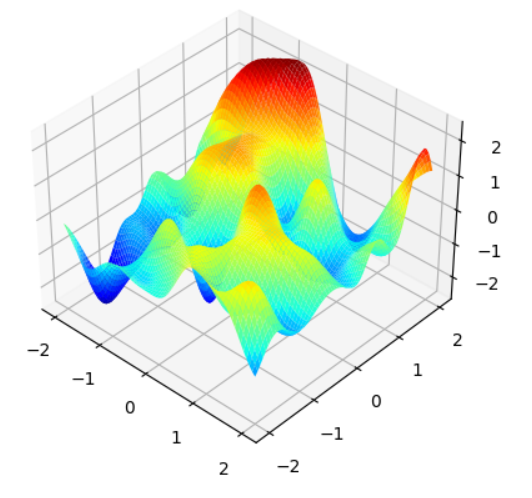} 
    \hspace{1cm}
    \includegraphics[scale=0.5]{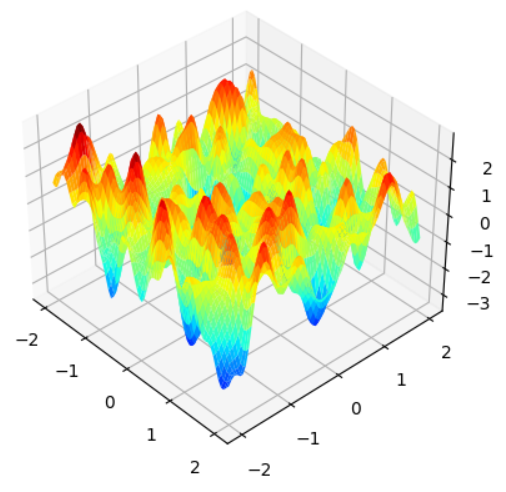} 

    \caption{Examples of Gaussian random fields. Top: 1D field, $\lambda=100$. Second: 1D field, $\lambda=10000$. Third: 2D field, $\lambda=100$. Bottom: 2D field, $\lambda=10000$.}
    \label{fig:GRF}
\end{figure}

Once a GRF is defined, we can use its \emph{maximal value} on $S$ in order to decide whether property $\Pi_{\text{YES}}$ or $\Pi_{\text{NO}}$ holds, as will be described next. 
We note that if $S$ is bounded (e.g. in a ball of fixed radius), the maximal value remains bounded even \emph{as the cardinality of $S$ tends to infinity}, in stark contrast to the FM hashing, for which the estimate of the cardinality (i.e. the inverse of the minimal value) will tend to infinity.

The proposed solution searches for the optimal values of the GRF parameter
$\lambda$, as well as that of a threshold $T$, which maximize the following for
a random $g$ drawn from the basic GRF (where $\max(g)$ always refers to the maximal value over $S$):
\begin{equation}\label{eq:opt1}
\Pr\left(\max(g) \geq T \mid  \text{case} \; \Pi_{\text{NO}} \right) -
\Pr\left(\max(g) \geq T \mid  \text{case} \; \Pi_{\text{YES}} \right)
\end{equation}

Noting that $\Pi_{\text{NO}}$ and $\Pi_{\text{YES}}$ are satisfied by a large class of sets $S$, in this work we follow a ``worst case scenario'', and maximize the expression in Equation \eqref{eq:opt1} under the following assumptions:

\begin{itemize}
    \item Case $\Pi_{\text{YES}}$: $S$ consists of $k_1$ balls of radius $\varepsilon$, i.e. it is \emph{maximal} (with respect to set containment) among all sets satisfying $\Pi_{\text{YES}}$. 
    Further, we assume no lower bound on the pairwise distances between these $k_1$ balls.
    \item Case $\Pi_{\text{NO}}$: $S$ consists of exactly $k_2$ vectors, which are equidistant from each other, all distance equal to $\delta$.
    Since we seek a general bound on  the field's maximal value, we will ignore the fact that this condition cannot always hold; for example, there are no 4 equidistant vectors in $\R^2$.
\end{itemize}

Given the optimal parameter values, the decision whether $\Pi_{\text{YES}}$ or $\Pi_{\text{NO}}$ holds proceeds by estimating Equation \eqref{eq:opt1},  by drawing a fixed number of random fields, and computing the empirical probability of their maxima to exceed $T$. 
Note that if $S$ is distributed among nodes, it suffices for each node to send an estimate of the maximum of the GRF on its part of $S$ to a coordinator node, which then simply takes the maximum of these local maxima; thus the algorithm is very appropriate for a distributed and/or dynamic setting, as the communication overhead does not depend on the cardinality of $S$.

There remains the question of how to compute the two probabilities
in Equation \eqref{eq:opt1}, which we address next.
One technical tool we use is the Slepian Lemma, for GRFs. 
\begin{theorem} The Slepian Lemma \cite{slepian}: let\\$(X_1 \ldots X_n)$ and
    $(Y_1 \ldots Y_n)$ be normal random variables with zero expectation, and further
    satisfying $E[X_k^2] = E[Y_k^2]$ for every $k$, and
    $E[Y_kY_{\ell}] \geq E[X_kX_{\ell}]$ for every $k,\ell$. 
    Then, for every $T$, $\Pr(\max{X_k} \geq T) \geq \Pr(\max{Y_k} \geq T)$.
\end{theorem}
That is, if the expectation and variance of $X_k,Y_k$ are equal for every $k$,
and all $X_k,X_{\ell}$ pairs are ``less correlated'' than $Y_k,Y_{\ell}$, then the probability of $\max{X_k}$ to exceed $T$ is larger than that that of $\max{Y_k}$ to exceed $T$, for every $T$. Surprisingly, this intuitive result is \emph{not} true for
all types of random variables.  

\subsection{Probability for the Maximum of the GRF to Exceed \texorpdfstring{$T$}{T} on a set Containing \texorpdfstring{$k_2$}{k2} vectors with all Mutual Distances \texorpdfstring{$\geq \delta$}{> d}.}
This problem reduces to that of computing the so-called exceedance probability for a $k_2$-dimensional normal distribution with zero expectation and a covariance matrix $C$ satisfying $C_{k,k}=1$ for all $k$, and $k \neq \ell \rightarrow C_{k,\ell} \leq \exp\left(-\lambda \delta^2\right)$. 
The following theorem allows to bound this probability from below by
solving for the case $C_{k,k}=1, C_{k \neq \ell}=\exp\left(-\lambda \delta^2\right)$: \\

Due to the Slepian Lemma, we need only to compute the probability for
the maximum to exceed $T$ for a normal vector $(X_1 \ldots X_{k_2})$ with
zero expectation and covariance matrix
\(C_{k,k}=1, C_{k \neq \ell} =\exp\left(-\lambda \delta^2\right)\). 
This probability can be computed as follows: denote $\exp\left(-\lambda \delta^2\right)$  by $\rho$. 
The random vector $(X_1 \ldots X_{k_2})$ can be modeled by
\[x_i = \sqrt{\rho}N_0+\sqrt{1-\rho}N_i, i=1 \dots k_2, \]
where $N_0,N_1 \ldots N_{k_2}$ are i.i.d standard normal random variables (it is easy to verify that each $X_i$ is a standard normal random variable, and that \mbox{$i \neq j \rightarrow \Cov\left(X_i,X_j\right) = \rho$)}.
For any $T_1$, the probability density of \mbox{$\max{\{\sqrt{1-\rho}N_1, \ldots, \sqrt{1-\rho}N_{k_2}\}}$} to equal $T_1$ is given by
\[
\frac{\sqrt{2}\, {\mathrm e}^{-\frac{\mathit{T_1}^{2}}{2 \left(1-\rho \right)}}}{2 \sqrt{\pi}\, \sqrt{1-\rho}}
\left(\frac{\mathrm{erf}\! \left(\frac{\sqrt{2}\, \mathit{T_1}}{2 \sqrt{1-\rho}}\right)}{2}\right)^{k_2-1},
\]
and the probability for the maximum of $x_i$ to exceed $T$ is therefore given by

\begin{eqnarray*}
\int_{-\infty}^{\infty} \left(\frac{\sqrt{2}\, {\mathrm e}^{-\frac{x^{2}}{2 \rho}}}{2 \sqrt{\pi}\, \sqrt{\rho}} \right) \cdot \,
\frac{\sqrt{2}\, {\mathrm e}^{-\frac{\mathit{T-x}^{2}}{2 \left(1-\rho \right)}}}{2 \sqrt{\pi}\, \sqrt{1-\rho}}
\left(\frac{\mathrm{erf}\! \left(\frac{\sqrt{2}\, \mathit{T-x}}{2 \sqrt{1-\rho}}\right)}{2}\right)^{k_2-1}\!\!\!\!\! \!\!\!\!\!dx ,
\end{eqnarray*}
which can be evaluated numerically.

\subsection{Probability for the Maximum of the GRF to Exceed  \texorpdfstring{$T$}{T} on  \texorpdfstring{$k_1$}{k1} Balls with Radius  \texorpdfstring{$\varepsilon$}{eps} each}
There are no known closed-form expressions for this probability, even in the one-dimensional case (i.e. an interval), and the probability is computed using a sampling approach, or approximations \cite{general_expression_grf_max_jean_mario,numerical_bounds_cecile}. 
In Section \ref{section:efficient_random_fields} we define a different field, for which there is a closed-form expression for the exceedance probability over a ball in any dimension. 
In order to derive an absolute bound, we treat the field's values on the $k_1$ balls as independent, which produces the bound  $1-\left(1-\Pr\left(\text{exceedance for a single ball}\right)\right)^{k_1}$.

\subsection{Examples for Optimal Parameters \texorpdfstring{$T,\lambda$}{T, l}}
Next follow some examples for choosing the optimal hyper-parameters $T,\lambda$, including a case in which the proposed paradigm fails, i.e. no $\lambda, T$ exist such that the exceedance probability for "non-clusterable" ($\Pi_{\text{NO}}$) is higher
than for "clusterable"  ($\Pi_{\text{YES}}$).


\paragraph{ \bf{Example 1}}
Here, we seek to differentiate between case $\Pi_{\text{YES}}$ being one cluster of size 0.01 (i.e. $k_1=1, \eps=0.005$, and case $\Pi_{\text{NO}}$ corresponding to a set which contains two vectors at least 0.02 apart ($k_2=2,\delta=0.02$), in dimension $d=1$. 
In Figure ~\ref{fig:optimal-T-and-lambda-for-1D}, the values of Equation \eqref{eq:opt1} are plotted for a range of $\lambda, T$. 
The optimal parameters are obtained at the highest point on the surface. 

Looking at $\lambda$ cross-sections of the 2D function in Figure \ref{fig:optimal-T-and-lambda-for-1D} allows an intuitive understanding of the dependence on $\lambda$.
If $\lambda$ is very small, then the covariance between $f(x), f(y)$, $\exp(-\lambda(x-y)^2)$, is very close to 1, hence the field is nearly constant, and its maximum is nearly identical for cases $\Pi_{\text{YES}}$ and $\Pi_{\text{NO}}$.
At the other extreme, where $\lambda$ is very large, the field's behavior approaches that of white noise, and the maximum is nearly independent of the extent of $S$ (while it does depend on the number of vectors in $S$).
For ``good'' $\lambda$ values, e.g. 3000, the correlation for case $\Pi_{\text{YES}}$ is bounded from below by $\exp(-3000 \cdot 0.01^2)=0.741$, while for case $\Pi_{\text{NO}}$ it reaches $\exp(-3000 \cdot 0.02^2)=0.301$.
This difference results in the exceedance probability of the field in case $\Pi_{\text{NO}}$ being higher than in case $\Pi_{\text{YES}}$, even when the number of vectors in case $\Pi_{\text{YES}}$ is unbounded -- as we compute the exceedance probability for a field on an interval of length 0.01. 

\paragraph {\bf {Example 2}}
There are however cases in which the proposed paradigm can fail, as there are no values of $\lambda, T$ for which Eq. \eqref{eq:opt1} obtains a positive value.
This is due to the fact that for case $\Pi_{\text{YES}}$, the exceedance probability must be bounded on a continuous set, as opposed to case $\Pi_{\text{NO}}$, where the field can only take values at a (typically small) finite set.
Such an example is when $\Pi_{\text{YES}}$ is defined by $k_1=3, \eps=0.05$, and $\Pi_{\text{NO}}$ by $k_2=4, \delta=0.1$.  
Since the difference of the exceedence probabilities (Equation \eqref{eq:opt1}) is always negative (Figure \ref{fig:example-fail}), there are no parameters $\lambda,T$ which can be applied.

\begin{figure}
    \centering
    \begin{minipage}{0.5\textwidth}
        \centering
        \includegraphics[scale=0.55, trim={0 0 0 2cm},clip]{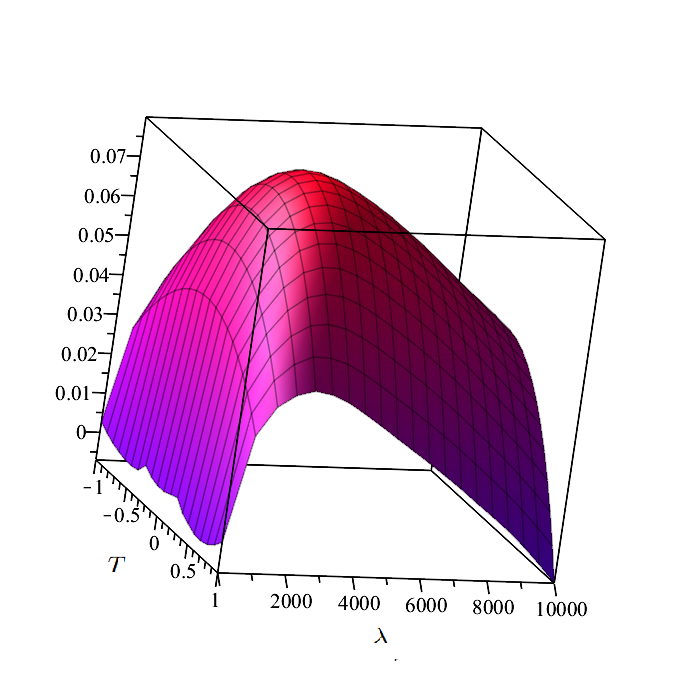} 
        \caption{An example of optimal $\lambda, T$ parameters}
        \label{fig:optimal-T-and-lambda-for-1D}
    \end{minipage}\hfill
    \begin{minipage}{0.5\textwidth}
        \centering
        \vspace{0.3cm}
        \includegraphics[scale=0.55, trim={0 0 0 1.5cm},clip]{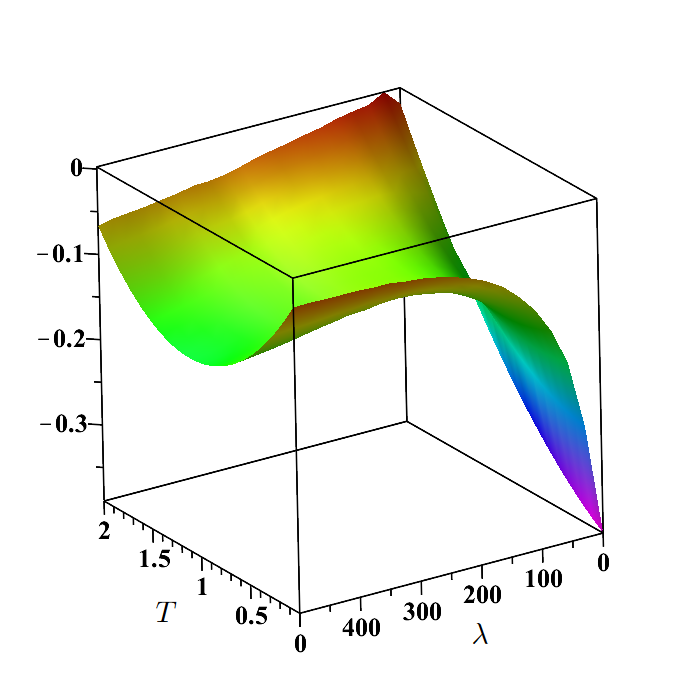} 
        \caption{An example in which the proposed paradigm fails}
        \label{fig:example-fail}
    \end{minipage}
\end{figure}

\subsection{Constructing the Gaussian Random Field} \label{sub:constructing}
There are no closed-form expressions for constructing a GRF, even in dimension 1. 
Here we propose an efficient approximation, which can be extended to higher dimensions.
The field is restricted to a finite interval, which without loss of generality we assume to be the unit interval $[0,1]$.

 Define 
\begin{equation}\label{eq:grf_aprox}
h(x) = \displaystyle \sum_{k=0}^n N_k(0,1)w_k \cos(\pi k (x-t)), \;\; t \sim U[-1,1]
\end{equation}
Since $h$ is the sum of normal variables, it is normal.
Also, as can be verified by direct integration, for any $x,y \in [0,1]$ the following holds:
$$\Cov\left(h(x),h(y)\right) = \frac{1}{2} \displaystyle \sum_{k=0}^n w_k^2 \cos\left(\pi k (x-y)\right) $$
Recall that for a GRF with parameter $\lambda$, 
$\Cov(h(x),h(y)) = \exp\left(-\lambda (x-y)^2\right)$. 
Therefore, we seek a valid approximation on the interval $[0,1]$:
\[
\exp\left(-\lambda t^2\right) \approx \sum_{k=0}^n \beta_k\cos(\pi kt) \]
Such an approximation can easily be found -- it is a truncated Fourier expansion
of the Gaussian. 
However, there is a caveat -- the Fourier coefficients are not all positive,
hence cannot be expressed as $w_k^2$.
While an optimal approximation using only positive coefficients can be sought, for our purpose it is  enough to simply omit the negative coefficients from the expansion. Since over the continuous domain,
the Fourier transform of a Gaussian is also a Gaussian,  all its coefficients are
positive.

This idea can be extended to higher dimensions $d$; alas, the size of the corresponding approximation increases with $d$ and $\lambda$ increases . The reason for the dependence on $\lambda$ is that more Fourier coefficients are required to approximate a Gaussian with a small standard deviation than with a large one. The number of terms required to approximate a $d$-dimensional Gaussian behaves as $\mathcal{O}(\lambda^{d/2})$.

\begin{figure}
    \centering
    \begin{minipage}{0.45\textwidth}
        \centering
        \includegraphics[scale=0.33]{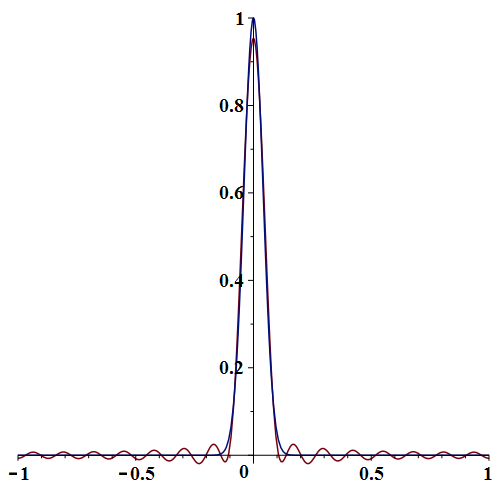} 
        \caption{The covariance of a 1D GRF with $\lambda=300$ (blue) and the covariance function of an approximation with 15 Fourier elements (red).}
        \label{fig:approx1}
    \end{minipage}\hfill
    \begin{minipage}{0.45\textwidth}
        \centering
        \includegraphics[scale=0.33, trim={0 0.2cm 0 0}, clip]{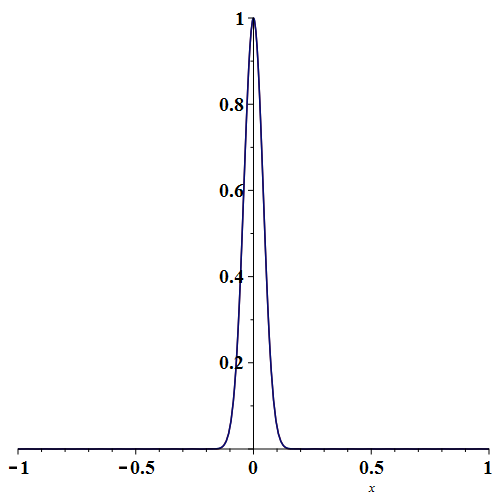} 
        \caption{The covariance of a 1D GRF with $\lambda=300$ (blue) and the covariance function of an approximation with 25 Fourier elements (red).}
         \label{fig:approx2}
    \end{minipage}
\end{figure}

\subsection{Algorithm and Analysis}
We now present our algorithm for distinguishing between the two cases, $\Pi_{\text{YES}}$ and  $\Pi_{\text{NO}}$. 
We assume that a random field was chosen, based on the parameters $\varepsilon, \delta$ (see Sections 
~\ref{section:efficient_random_fields}, ~\ref{section:lower_bound_optimal_kernels}, and ~\ref{section:higher_dimenstion}).
\\
\begin{mdframed}[nobreak=true,align=center]
\label{main_alg:step1}
{\bf Input.} Parameters: $\varepsilon, \delta \in \R_{+}$ and $k_1, k_2 \in \N$.
\begin{enumerate} 
    \item Given the parameters $\varepsilon,\delta, k_1$ and $k_2$ , find the $T$ and  $\lambda$ \footnote{This step is done numerically. 
This is a main limitation of our method, and we leave a more rigorous analysis of this step for future research.} that maximize Equation \eqref{eq:opt1}:
    \begin{equation} \nonumber
        \Pr\left(\max(g) \geq T \mid  \text{case} \; \Pi_{\text{NO}} \right) -
        \Pr\left(\max(g) \geq T \mid  \text{case} \; \Pi_{\text{YES}} \right)\text{.}
    \end{equation}
    Denote:
    \begin{itemize}
        \item $\calC=\Pr\left(\max(g) \geq T \mid  \text{case} \; \Pi_{\text{YES}} \right)$
        \item $\M = \Pr\left(\max(g) \geq T \mid  \text{case} \; \Pi_{\text{NO}} \right)$
    \end{itemize}
    
    \item If no $\lambda, T$ exist for which the equation above is strictly positive, return {\bf FAIL}
    \item Let $\Pp \leftarrow \Pr \left(\max(g)>T\right)$ over $S$\footnote{The probability is computed empirically by generating a number of fields.}, for a set of vectors $S \subseteq \R^d$.
    \item If $\Pp < \frac{\M+\calC}{2}$ return $\Pi_{\text{YES}}$; otherwise, output $S \in \Pi_{\text{NO}}$.

\end{enumerate}
\end{mdframed}

Using standard probability techniques, the success rate of the algorithm can be estimated.

\subsection{Distributed Systems and Privacy}
As for the FM algorithm, the approach described here for deciding between $\Pi_{\text{YES}}$ and $\Pi_{\text{NO}}$ is very suitable in a distributed, dynamic system: since the outcome depends only on the maximal value of the field, it requires only a constant volume of communication per node, and is also easy to adapt to the case of distributed data streams. The nodes merely need to agree on a seed for generating the random field (as the same fields must, of course, be applied on the entire data).

Interestingly, the above observation still allows a measure of privacy.
Assume for example a system with two nodes $N_1,N_2$.
If the maximal values on the (identical) fields used by the nodes are well-correlated, $N_2$ can infer on $N_1$'s data (for example, if $\Pi_{\text{YES}}$ holds, $N_2$ knows that $N_1$'s data can be covered by balls which also cover its own data).
However, if the maxima are uncorrelated, such inference is much more difficult and partial, as the average values of the GRF are invariant to data rotation and translation; therefore $N_2$ can only infer on $N_1$'s data up to rotation and translation. We leave further analysis of this subject for future research.
\subsection{Dimensionality Reduction}
The proposed algorithm is susceptible to curse of dimensionality, as the maximal value of the GRF on a ball with a fixed radius increases with the dimension, hence the exceedance probability for $\Pi_{\text{YES}}$ increases with $d$, while for $\Pi_{\text{NO}}$ it does not. 
This is also true for the other random fields we define (Sections \ref{section:efficient_random_fields}, \ref{section:higher_dimenstion}).
As opposed to other problems, ours cannot be directly remedied by applying random projections to low dimensions (in the spirit of the famous {\em Johnson-Lindenstrauss Lemma}), since, while preserving the {\em ratio} between distance, the {\em absolute} distances decrease;
however, the projection of a ball is a ball {\em with the same radius}.
So while the dimension decreases, so does $\delta$, making condition $\Pi_{\text{NO}}$ more difficult to satisfy. 
The trade-off between reducing $d$ and reducing $\delta$ is an interesting topic for future research.

\section{Efficient Construction of Random Fields} \label{section:efficient_random_fields}

While the GRF kernel is very suitable to our problem, it is computationally expensive
to construct, especially in high-dimensions, and its applications are typically restricted to dimension 3 and below \cite{brain_fields,CreaseyLang+2018+1+11}.
Further, it is untrivial  to compute its exceedance probability on a ball; even for the 1D case (interval), there are no closed-form expressions. 

We start the analysis of other fields with the well-known method of Locally Sensitive Hashing (LSH) \cite{DBLP:conf/compgeom/DatarIIM04}. 
The hash function, which accepts real valued vectors as inputs, is defined (up to a scaling factor) by:
\begin{equation}
    h(x)= \left \lfloor {a\cdot x +b}\right \rfloor, \; a_i \sim N(0,\sigma), b \sim U[0,1]
\end{equation}
where the normal distribution $N(0,\sigma)$ can be replaced by another \emph{Stable Distribution}. 

In order to compare this hashing scheme to ours, we estimate the mean squared distance between the hashed images of two vectors $x,y \in \R ^d$. 
We ignore the rounding operation (which is a necessity imposed by the requirement to compare hash values). 
Direct integration yields:
%
\begin{eqnarray*} \nonumber
& & \E[(h(x)-h(y))^2]= {(2\pi\sigma^2)^{-\frac{d}{2}}} \cdot\\ \nonumber
& & \int_{-\infty}^{\infty}\ldots\int_{-\infty}^{\infty} \exp\left(-\frac{\|a\|^2}{2\sigma^2}\right)\left(a\cdot x - a\cdot y\right)^2 da_1 \ldots da_d \\ \nonumber & & = \sigma^2\|x-y\|^2 \nonumber
\end{eqnarray*}
It's also easy to verify that $\E[f(x)]=0$, and $E[f(x)f(y)]=x\cdot y$. 
Hence, the covariance of $h(x),h(y)$ does not tend to 0 when $\|x-y\|$ increases, as the GRF does; further, the maximal value of the hash function will not be related to the number of clusters, but to the diameter of the set. Hence, while LSH is highly successful in solving other problems (e.g. nearest neighbor), it is not suitable to our goal.
\subsection{Random Trigonometric Fields}
In Section \ref{sub:constructing}, we looked at a method to approximate a GRF using sums of trigonometric functions, with random coefficients.
Not surprisingly, the larger $\lambda$ is, the more terms are required to approximate a field with parameter $\lambda$, which satisfies $\Cov\left(f(x),f(y)\right)=\exp(-\lambda\|x-y\|^2)$, as the fields become "rougher" when $\lambda$ increases.
It is natural to ask whether similar fields can be obtained while using a smaller number of trigonometric terms; it turns out that this can be achieved with one term, if the choice not of the coefficients, but the argument of the trigonometric functions (i.e. the frequency), is random \footnote{In the \emph{cosine field}, \cite{adler2009random}, the frequency is constant.}.

\begin{definition}[Random Sine Field]
  Let $x \in \R^d$. 
  Define a random field -- which will be referred to hereafter as the \emph{Random Sine Field}(RSF), with parameter $\sigma$, by:
\begin{equation}
    \label{sine_field}
    f(x) = \sin\left(a \cdot x + b\right)
\end{equation}

where $a \in \R^d$ is a random vector with components which are i.i.d normal
variables with zero expectation and $\sigma$ standard deviation, and $b$ is uniformly
distributed between $-\pi$ and $\pi$.
\end{definition}
  
\begin{lemma} \label{lemma-covariance-sine-field}
\leavevmode
\begin{enumerate}
  \item  for every $x\in \R^d$, $f(x)$ is a random variable  with zero expectation and a variance equal to $\frac{1}{2}$.
  \item The covariance of $f(x)$ and $f(y)$ is equal to:
$$\frac{1}{2}\exp\left(-\frac{\sigma^2\|x-y\|^2}{2}\right).$$
\end{enumerate}
\end{lemma}
  
\begin{proof}
We start with the expectation: 

\begin{eqnarray*} \nonumber
& &  E[f(x)]=\frac{1}{2\pi}  \int_{-\infty}^{\infty}\ldots\int_{-\infty}^{\infty} \frac{1}{(2\pi\sigma^2)^\frac{d}{2}} \nonumber \\
& & \int_{-\pi}^{\pi}\exp\left(-\frac{\|a\|^2}{2\sigma^2}\right)\sin(a \cdot x +b) db da_1 \ldots da_d = 0 \nonumber
\end{eqnarray*}
as follows immediately by first integrating over $b$.
  
For the variance, the proof also follows by direct
integration:
\begin{eqnarray*} \nonumber
 & & E[f(x)f(y)]=\frac{1}{2\pi}  \int_{-\infty}^{\infty}\ldots\int_{-\infty}^{\infty} \frac{1}{(2\pi\sigma^2)^\frac{d}{2}} \nonumber \\ 
 & & \int_{-\pi}^{\pi}\exp\left(-\frac{\|a\|^2}{2\sigma^2}\right)\sin(a \cdot x +b)\sin(a \cdot y +b) db da_1 \ldots da_d \nonumber \\
 & & = \frac{1}{2}\exp\left(\frac{-\sigma^2\|x-y\|^2}{2}\right) \nonumber
  \end{eqnarray*}
  
\end{proof}

The drawback of this field is that, as opposed to the basic GRF, it does not
satisfy the property that, for each $x_1 \dots x_n$, the vector $\left(f(x_1) \ldots f(x_n)\right)$ obeys a Gaussian distribution. 
Therefore, we cannot a-priori use the Slepian Lemma \cite{slepian}, which guarantees that, as the vectors in case $\Pi_{\text{NO}}$ move farther apart, the exceedance probability of the maximum of the field increases. 
However, we were able to prove an analogous result for the sine field:

\begin{lemma} \label{Slepain-for-sine}
     The probability of the RSF to exceed any given value on a set of vectors, is monotonically increasing in all the distances between the vectors \footnote{We assume that the vectors are normalized.}.
\end{lemma}

The proof of the Lemma is in section \ref{section:slepian-for-sine}. 

 \smallskip
 In addition to the very low computational complexity of generating the RSF -- even in high dimensions --  it has another  advantage:  for a ball $B_{\eps}$ of radius $\eps$ and a threshold $T$,
 there exists an explicit, simple expression for the probability of the maximum of the field over
$B_{\eps}$ to exceed $T$.

\begin{lemma}
    \label{lemma:proof-exceedance}
      The probability of the RSF to exceed $T$ on $B_{\eps}$
      is given by 
      \begin{equation}
        \frac{\pi-2\arcsin(T)+2\eps\sigma\sqrt{d}}{2\pi}
      \end{equation}
      where it is tacitly assumed that the numerator is $\leq 2\pi$.
\end{lemma}

 \smallskip      

\begin{proof}
Clearly the probability of a random angle $\theta$ to satisfy $\sin(\theta) \geq T$
is $\frac{\pi - 2\arcsin(T)}{2\pi}$. It remains to compute the probability that the
image of $B_{\eps}$ under a random hash function intersects the interval of
length $\frac{\pi - 2\arcsin(T)}{2\pi}$ on the unit circle in which $\sin(\theta) \geq T$.
The image of $B_{\eps}$ under $f(x) = a \cdot x $ is an interval of average
length $2\eps\sigma\sqrt{d}$, as its endpoints are the inner products of $a$ with
the vectors on $B_{\eps}$ in the directions of $a$ and $-a$.
The average length
of a $d$-vector whose coordinates are all drawn from $N(0,\sigma)$ is 
$\sigma\sqrt{d}$. The probability for both intervals to intersect is the sum of their lengths,
hence the result follows. The geometric idea behind the proof is provided in Fig. \ref{fig:proof-exceedance-1}.
\end{proof}
    
    \begin{figure}[b!]
    \centering
  \hspace*{-0mm}  \includegraphics[scale=0.36]{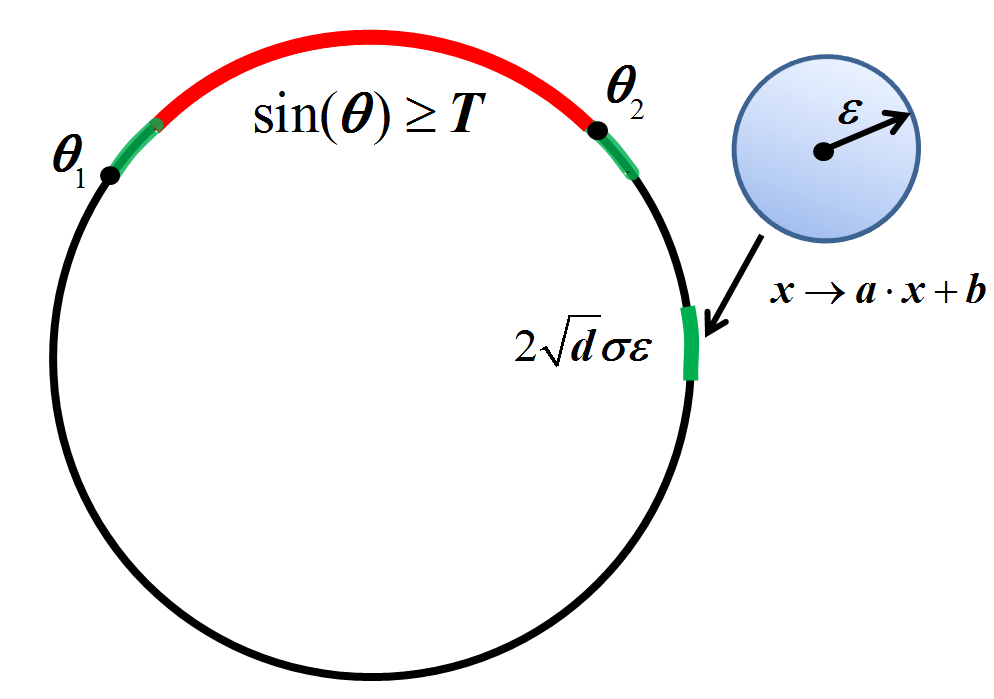}
    \caption{Sketch for the proof of Lemma \ref{lemma:proof-exceedance}. The red "dome" represents the range of angles $\theta$ for which $\sin(\theta) \geq T$. The $d$-dimensional ball with radius $\epsilon$ is projected onto a range of  $\theta$ of length  $2\eps\sigma\sqrt{d}$ (green). For the two intervals to intersect, the starting point of the "green" interval has to be between $\theta_1$ and $\theta_2$.} \label{fig:proof-exceedance-1}
\end{figure}

Choosing optimal values for the parameters $T,\sigma$ proceeds as for the GRF (Section ~\ref{section:hash_function_clustring}), however they are  easier to compute, since computing Eq. \eqref{eq:opt1} is simpler.

Fig. ~\ref{fig:RSF-parameters} provides a simple example, which captures the intuition behind the ability of the RSF to distinguish conditions $\yes$ and $\no$: for small $\sigma$, the field is nearly constant, and no differentiation between the aforementioned conditions is possible. When $\sigma$ becomes too large, the field oscillates rapidly, and always obtains the maximal value of 1 on a set of size $\eps$, causing $Pr(\max(g) \geq T \mid \yes)$ to be larger than $Pr(\max(g) \geq T \mid \no)$, as $\no$ contains finitely many vectors.
The optimal value is obtained between these extreme values. 
This echoes the choice of the optimal $\lambda$ for the GRF (Section ~\ref{section:hash_function_clustring}); for small $\lambda$ the field is nearly constant, but as $\lambda$ increases, the field changes more rapidly, until the expected maximum for $\yes$ will be larger than for $\no$. 
Intuitively, we seek values of $\lambda$ and $\sigma$ for which the range of field values for $\no$ is the largest w.r.t to $\yes$.
 
   \begin{figure}[t!]
    \centering
    \includegraphics[scale=0.6]{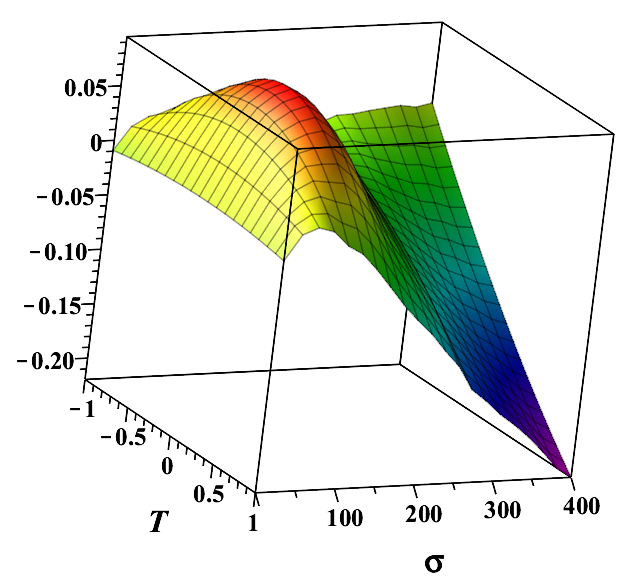}
    \caption{An example of choosing the optimal $T,\sigma$ for the same parameters as in Fig.
    \ref{fig:optimal-T-and-lambda-for-1D} ($d=1,k_1=1,k_2=2,\delta=0.02,\eps=0.01$). As in  Fig.
    \ref{fig:optimal-T-and-lambda-for-1D}, the difference between the exceedance probabilities of conditions $\no$ and $\yes$ is plotted, relative to a range of values for $T,\sigma$. Note that in this case, RSF performs on par with the GRF, as the maximal difference between $Pr(\no)$ and $Pr(\yes)$ is similar.} \label{fig:RSF-parameters}
\end{figure}

   \subsection{Random Fields and Stable Distributions}
   When defining the RSF, we randomly chose the vector $a$ such that each of its components obeys a normal distribution. The choice of the distribution over $a$ must be judicious, as demonstrated by the following example, which holds even in the case of a 1D field: define a field by $f(x)=sin(ax+b)$, where $b \sim U[-\pi,\pi]$ and $a \sim U[-A,A]$. It is easy to verify that  $\Cov(f(x),f(y))=\frac{\sin(A(x-y)}{2(x-y)}$ (i.e. a {\em sinc} function at $x-y$). This is inappropriate  for our goal, as the covariance is not positive nor monotonically decreasing in $|x-y|$.
   
   However, just as the normal distribution was used to construct the RSF in Definition \ref{sine_field},
   we can use another well-known example of a stable distribution for the coefficients of the vector $a$,
   which defined the field at $x$ by $\sin(a \cdot x + b)$:
   $\Pr(a_i)=\frac{k}{\pi(1+k^2a_i^2)}$ (where $b$ is uniformly distributed in $[-\pi, \pi]$, and $k>0$).
   Allowing a slight abuse of terminology, we shall refer to it as the {\em Stable Field}.
   A somewhat more elaborate computation (which will not be repeated here) than for the case of the RSF, yields that in this case, the covariance matrix of the field is given (up to a multiplicative constant) by:
   \[\Cov(f(x),f(y))=\exp\left(-\frac{\|x-y\|_1}{k}\right) \]
   while this field is also viable for our problem (it is stationary, and the covariance monotonically decreases with the $L_1$ distance),
   its performance will be inferior to that of the GRF and RSF;
   this is due to the fact that the decrease in the covariance's values from "distance $\eps$" to "distance $\delta$" will be slower.
   If e.g. $\delta=10\eps$, and we set the covariance to be 0.99 for distance $\eps$, then for distance $\delta$ it will be 0.368 for the RSF and 0.905 for the stable field.
   Still, an interesting topic for future work is the relation between stable distributions and random fields, and their application to the problem studied here.
   
   Figs. \ref{fig:1Dstable},\ref{fig:2Dstable} include one and two dimensional examples of the stable field.
   Note that these are {\em not} samples of the random field, which is just a sinusoidal function, but a visualization of the covariance. These fields tend to be more "jittery" than the GRF and RSF.

\begin{figure}
    \centering
    \begin{minipage}{0.45\textwidth}
        \centering
        \includegraphics[width=8cm,height=5cm]{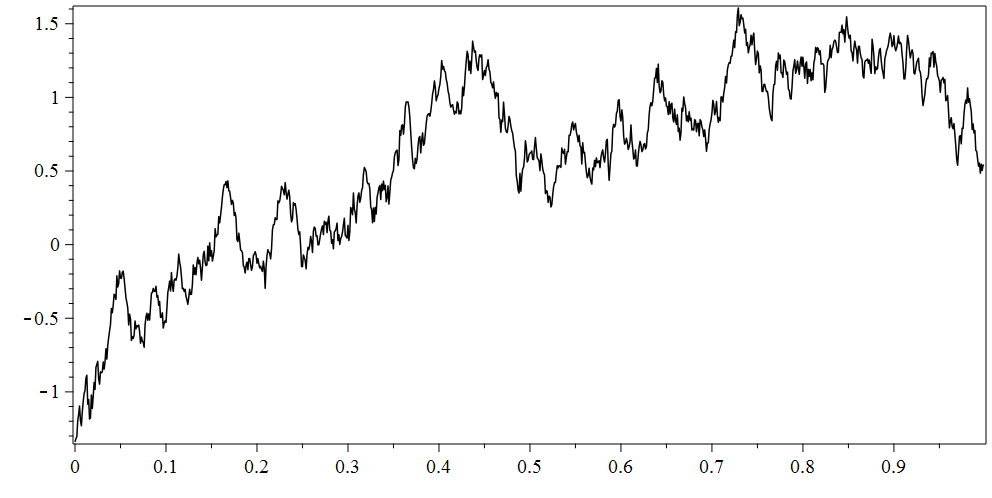}
        \caption{An example of a 1D stable field with  $k=10$ }
        \label{fig:1Dstable}
    \end{minipage}\hfill
    \begin{minipage}{0.45\textwidth}
        \centering
        \includegraphics[scale=0.7]{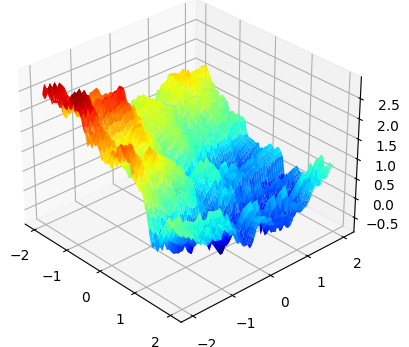}
        \caption{An example of a 2D stable field with   $k=10$ }
        \label{fig:2Dstable}
    \end{minipage}
\end{figure}

\section{Slepian-type result for RSF} \label{section:slepian-for-sine}
In this section we prove an analogous version of the Slepian Lemma for the sine random field. 
Recall that the ``sine field'' is defined by
\[
R(u) = \sin(a \cdot u + b), a_i \sim N(0,\sigma), b \sim U[-\pi,\pi]
\]
(assume hereafter $\sigma=1$, the proof remains the same for every $\sigma$).

As a direct computation proves \ref{lemma-covariance-sine-field}, the covariance of $R(x),R(y)$ equals $\exp(-||x-y||^2)$. 
The decision whether properties $\Pi_{\text{NO}}$ or $\Pi_{\text{YES}}$ hold for a set of
points $S$ depends only on the exceedance probability: given $T$,
the probability that the maximum value $R()$ attains on $S$ is at
least $T$.

We next prove Lemma \ref{Slepain-for-sine}, which states that as the
points of $S$ move farther apart, the expectation of the maximum of
the sine field increases.

Let us start with two points, $u_1, u_2 \in {\R}^d$. 
Assume $||u_i||=1$.
We then have, from standard results about Gaussian random variables:
\[
R(u_i) = \sin(x_i + b), x_i \sim N(0,1), \Cov(x_1,x_2) = u_1 \cdot u_2 \]
Since $u_1 \cdot u_2 = \frac{1}{2}(||u_1||^2+||u_2||^2-||u_1 - u_2||^2)$ and
$||u_i||=1$, $\Cov(x_1,x_2)$ decreases when $||u_1 - u_2||$ increases.

According to \cite{slepian1963zeros}, when applied to
$f(x,y)=\int\limits_{-\pi}\limits^{\pi}\max(\sin(x+b),\sin(y+b))db$, it suffices to prove that

\begin{equation} \label{eq:slepian_integral}
  \int\limits_{-\infty}^{\infty}\int\limits_{-\infty}^{\infty}
  p(x,y)\frac{\partial^2 f}{\partial x \partial y}dxdy \leq 0
\end{equation}

Next, we look at $\frac{\partial^2 f}{\partial x \partial y}$. 
This is a little tricky, as the mixed derivative of $\max(u,v)$ is 0 when
$u \neq v$, and does not exist when $u=v$. 
However, standard arguments allow to relate to it as $-\frac{1}{2}\delta(u,v)$. 
Inserting the derivative into the integral which defines $f(x,y)$, we have
\[
\frac{\partial^2 f}{\partial x \partial y} =
\int\limits_{-\pi}\limits^{\pi}\cos(x+b)\cos(y+b)
\frac{\partial^2 \max(u,v)}{\partial u \partial v}(u=\sin(x+b),v=\sin(y+b))db 
\]

Therefore, the only points which contribute to the integrand are the values $b$ for which $\sin(x+b)=\sin(y+b)$. When $x=y+2\pi k$, this holds for every $b \in [-\pi,\pi]$, hence the integrand is a delta function at every point. 
In addition, the integrand contains a $\cos(x+b)$ factor (and not $\cos(x+b)\cos(y+b)$, as one of the cosines is canceled by the inner derivative of the delta function). 
This means that along the lines $y=x+2\pi k$, $f(x,y)$ is negative, and its value tends to infinity w.r.t relative to the value at other points; therefore, the maximal value of the integral in Eq. \ref{eq:slepian_integral} is obtained when
the values of $p(x,y)$ obtain their smallest possible value along the ``ridge'' $y=x$, i.e. when the covariance of $p(x,y)$ approaches -1, and $p(x,y)$ approaches a degenerate Gaussian whose which equals 1 when $y=-x$, and 0 otherwise. 

For this Gaussian, the integral in Eq. \ref{eq:slepian_integral}
reduces to an integral over the line $y=-x$. This one-dimensional
integral can be readily computed by distinguishing between the case
$x=k \pi, y = -k \pi$ (in which the integrand is a delta function at
every point), and the other cases; its value turns out to be 0. Since
this is the maximal value, the condition in Eq. \ref{eq:slepian_integral} is fulfilled,
and the expectation of the maximum indeed increases when the covariance
of $x,y$ decreases. The general case (more than two random variables) follows
immediately, since we need only look at all the partial derivatives involving all pairs of variables.

\section{Lower Bounds and Optimal Kernels} \label{section:lower_bound_optimal_kernels}
We now look at a more general case, in which the covariance matrix of the field is defined by any \emph{positive definite kernel} $k()$, that is, a function from $\R^d$ to $\R$ with the property that, for every set of vectors $x_1 \ldots x_n$, the matrix $C_{i,j}=k\left(\|x_i-x_j\|\right)$ is positive definite. 
This allows to apply more general kernels, whose covariance matrix $C_{i,j}$ can be different from the one of the basic GRF, for which $k(t)=\exp(-\lambda \|t\|^2)$, and $C_{i,j}=\exp(-\lambda \|x_i-x_j\|^2)$. 
Many families of positive definite kernels are known, for example the {\em Matérn kernels} \cite{DBLP:books/lib/RasmussenW06}, and applications, for example in machine learning, can benefit from the choice of a suitable kernel \cite{DBLP:journals/jmlr/Genton01,pr8010024}. 
Here we deal with the case $d=1$, and in Sections ~\ref{section:higher_dimenstion}, ~\ref{section:numerical_eval} the general case is addressed.

As discussed in Section \ref{section:hash_function_clustring}, an optimal kernel $k()$ for our problem is one which, loosely speaking, treats vectors with distance smaller than $\varepsilon$ as ``nearly identical'', and vectors with distance larger than $\delta$ as ``nearly independent''. 
Ostensibly this can be achieved by applying a ``step function'' kernel
\[
k(t)=  \begin{cases} 1 & \text{ if } t \leq \alpha \\ 0 & \text{ if } t>\alpha  \end{cases} 
\]
for some $\varepsilon<\alpha<\delta$. 
Alas, this kernel is not positive definite, as can be seen even for a case with three vectors. Assume the vectors are $x_1,x_2,x_3$, and satisfy $\|x_1-x_2\|,\|x_2-x_3\|<\varepsilon$, and $\|x_2-x_3\|>\delta$. 
Then the matrix defined by $C_{i,j}=k(\|x_i-x_j\|)$ is equal to
\[ \left( \begin{array}{ccc}
1 & 1 & 0 \\
1 & 1 & 1 \\
0 & 1 & 1 \end{array} \right)\] 
which is not positive definite, as its eigenvalues are\\$1+\sqrt{2},1,1-\sqrt{2}<0$. Hence, it is impossible to realize it as the covariance matrix of any field.

We next formalize the requirements for an optimal kernel function $k()$, given the problem parameters $\varepsilon,\delta$. 
In order for the  problem to be well-defined, we normalize $k()$ by requiring $k(0)=1, k(\varepsilon)=c$. 
Since $\varepsilon$ is the allowed size of a ``tight'' cluster, and we aim for the field to map tight clusters to a very small range, we typically set $c$ to be very close to 1, e.g. 0.99.
We also require $k()$ to monotonically decrease, since the correlation between the field values at $x,y$ must decrease as $\|x-y\|$ increases. 
We also demand that $k()$ be non-negative, for the following reason: since $k(0)=1$, if there exists a $t_1$ such that $k(t_1)<0$, then by continuity there must be a value $t_2$ such that $0<t_2<t_1$, and $k(t_2)=0$. 
But then, from the Slepian Lemma,  
\[\Pr\{\max\{f(0),f(t_1)\} \geq T\} > \Pr\{\max\{f(0),f(t_2)\} \geq T\} \]
in contradiction to the requirement that the exceedance probability of the maximum increases as the vectors move farther apart.
\\~~
\\~~
All the requirements from $k()$ are now summarized:
\begin{enumerate}
\item Domain: $k(t)$ is defined in the range $0 \leq t \leq 1$.
\item Normalization: $k(0) = 1,k(\varepsilon) = c$.
\item $k()$ is positive definite.
\item $k()$ is positive.
  \item $k()$ is monotonically decreasing.
  \item Under conditions 1-5, $k(\delta)$ is minimal.
\end{enumerate}

Kernels that satisfy 3,4 above are often referred to as \emph{doubly positive},
and are of considerable interest in physics and applied mathematics \cite{dirac_fourier,efimov2016integral,positive_functions_positive_fourier_giraud,positivity_of_fourier_giraud2014}. 
From the celebrated Bochner Theorem \cite{positive_positivedefinite}, a kernel is doubly positive iff both it and its Fourier transform are positive. 
The rate of decay of such kernels was studied, unfortunately there do not exist concrete results for solving condition 6.

A large family of candidate kernels is provided by the following theorem,
due to Polya \cite{Polya}:
\begin{theorem}
Let $\phi:\R \rightarrow [0,\infty]$ be even, decreasing, and convex. 
Then it defines a positive definite kernel. 
\end{theorem}

Alas, the convexity constraint severely hampers the goal of obtaining a
minimal value of $k(\delta)$. 
To illustrate this, observe that for  $\varepsilon=0.01,\delta=0.1,c=0.99$, the value
of a convex kernel at 0.1 is bounded from below by 0.9, while the Gaussian kernel (known to be doubly positive), $\exp(-\lambda t^2)$, with $\lambda$ chosen to satisfy $k(0.01)=0.99$, satisfies $k(0.1) = 0.366$.

The question therefore is: can the value $k(\delta)$ obtained by the basic GRF (Gaussian kernel) be decreased further? To the best of our knowledge, there are no known bounds on the minimal value. 
We now describe two methods for bounding  $k(\delta)$. 
Such bounds are very important, as they allow us to determine how close a certain kernel is to the optimal one.

\subsection{Bounding \texorpdfstring{$k(\delta)$}{k(delta)} using positive definiteness of kernel matrices}
Let us define, for every non-negative integer $l$, $p_l = k(l\varepsilon)$. 
If $k()$ is a positive definite kernel, then for every $x_1,x_2 \ldots x_n$, the $n \times n$ matrix $M_{i,j} = k(\|x_i - x_j\|)$ is positive definite. 
This allows to derive many inequalities which are useful in bounding $k_l$. 
The most basic one is provided in the following simple lemma.
\begin{lemma}
  \label{lemma:k2}
  $p_2 \geq 2p_1^2 - 1$.
\end{lemma}

\begin{proof} taking $x_0 = 1, x_1 = \varepsilon, x_2 = 2\varepsilon$, the
  matrix
  \[
   M=
  \left[ {\begin{array}{ccc}
   1 & p_1 & p_2 \\
   p_1 & 1 & p_1 \\
   p_2 & p_1 & 1 \\
  \end{array} } \right]
  \]
must be positive definite, hence its determinant is positive. 
But $|M| = (1-p_2)(1+p_2 - 2p_1^2)$. 
Since $1-p_2 \geq 0$, the result follows.
\end{proof}
While useful, Lemma \ref{lemma:k2} does not always suffice to provide lower bounds
on the values of $p_l$ -- actually, not even $p_2$. 
In order to see this, assume $p_1 = 0.7$. 
Then, Lemma \ref{lemma:k2} provides the bound $p_2 \geq -0.02$. 
Since $p_2$ must be positive, this provides a lower bound of 0; and indeed, the corresponding kernel matrix
\[
\left[ {\begin{array}{ccc}
   1 & 0.7 & 0   \\
   0.7 & 1 & 0.7 \\
   0   & 0.7 & 1 \\
  \end{array} } \right]
\]
is positive definite. 
Now, the positivity and monotonicity constraints imply,
obviously, that $p_3 = 0$ as well. 
However, the corresponding kernel matrix, constructed using the points $0,\varepsilon,2\varepsilon,3\varepsilon$ equals
\[
\left[ {\begin{array}{cccc}
1 & 0.7 & 0 &0 \\
0.7 & 1 & 0.7 & 0\\
0   & 0.7 & 1 & 0.7\\
0   & 0   & 0.7& 1 \\  
\end{array} } \right]
\]
which is \emph{not} positive definite, having a strictly negative eigenvalue
($\approx -0.1326$).

Hence, the lower bound provided by Lemma \ref{lemma:k2} is non-optimal (unless
$\varepsilon > \frac{1}{3}$, in which case $3\varepsilon > 1$ and is not a valid point).

This non-optimality of this bound is due to its ``restricted view'', so to say;
it is derived using only the points $0,\varepsilon,2\varepsilon$ and their corresponding kernel matrices, while ignoring points with larger coordinates (e.g. $3\varepsilon$).

The lesson of this simple analysis is that better lower bounds on $p_l$
can be obtained by looking at larger kernel matrices, involving more and more points.
For example, let us look at a typical range of parameters: $\varepsilon=0.01, c=0.99$, and try to bound $k(0.1)$ from below. 
In principle, a bound can be found by solving the following
optimization problem:
\[
  \text{argmin } {p_{10}} \text{ subject to} \hspace*{2mm} A_i \succeq 0, i = 1,2... \ldots n.
\]
where $\succeq 0$ stands for positive definiteness, and $A_i$ is the kernel matrix
constructed using the values $p_l$ for $l=1 \ldots i$. 
Alas,that will result in very large kernel matrices, which renders the optimization
problem non-practical (for example, optimizing over determinants of
matrices which are larger than $10 \times 10$ is very difficult, due both to
the high degree of the resulting polynomials and their sheer size). 
Moreover, the large matrices encountered in the optimization process turn out to be extremely ill-conditioned.

A more tractable solution is to use a hierarchy of bounds. 
For example, $p_{10}$ can be bounded as follows: first, take $\varepsilon=0.01$, and bound $p_5$. 
Then look at  $\varepsilon = 0.05$, and bound $p_2$ (which is equivalent
to bounding $p_{10}$ for $\varepsilon=0.01$).
\subsection{Bounding \texorpdfstring{$p_5$}{p5}}
In order to bound $p_5$, we looked at the following combinations
of points: $\{0,\varepsilon,2\varepsilon\}$,$\{0,\varepsilon,2\varepsilon,3\varepsilon\}$,
$\{0,4\varepsilon,5\varepsilon\}$, $\{0,2\varepsilon,4\varepsilon\}$. 
The determinants of the resulting kernel matrices are:
\begin{align}
    & \left( { p_2}-1 \right)  \left( 2\,{{ p_1}}^{2}-{p_2}-1 \right), \nonumber\\
    &   \left( {{ p_1}}^{2}+2\,{ p_1}\,{ p_2}-{ p_1}\,{ p_3}+{{ 
    p_2}}^{2}-{ p_1}-{ p_3}-1 \right) \nonumber \\
    & \left( {{ p_1}}^{2}-2\,{ p_1
    }\,{ p_2}-{ p_1}\,{ p_3}+{{ p_2}}^{2}+{ p_1}+{ p_3}-1 \right),  \nonumber \\
    & 2\,{ p_1}\,{ p_4}\,{ p_5}-{{ p_1}}^{2}-{{ p_4}}^{2}-{{ p_5
    }}^{2}+1, \nonumber \\
    & \left( { p_4}-1 \right)  \left( 2\,{{ p_2}}^{2}-{ p_4}-1 \nonumber
     \right) 
\end{align}
The relative simplicity of these expressions (polynomials of degree at
most three or products of thereof), allowed a closed-form solution
using the Maple\textsuperscript{TM} software package, providing a lower bound of $0.756$
for $p_5$.

It remains to bound $p_2$, assuming $p_1=0.756$. The naive bound
$2p_1^2-1$ yields $0.15477$, however, when the points
$\{3\varepsilon,4\varepsilon\}$ and the corresponding kernel matrices
are also considered (that is, we bound $k(2\varepsilon)$ using also
the values $k(3\varepsilon),k(4\varepsilon)$), then a much better bound
is obtained, $0.284$.

We note here that, to the best of our knowledge, no better bounds exist;
however, ostensibly such bounds can be found, by refining the optimization
process. 

We next provide a different approach for bounding the kernel values,
using a discrete approximation to $k()$. This approach typically
produces better bounds,  and also allows to construct approximately optimal kernels.
\subsection{Bounding the kernel values using a discrete Fourier approximation}  \label{sec:bochner1d}  
In this approach, we derive a lower bound on $p_l$ by directly
optimizing over the kernel $k(t)$, while imposing requirements 1-6 described in
this section's first part.
Alas, optimization techniques such as the calculus of variations lead
to unsolvable equations (for example, the monotonicity constraint is
equivalent to imposing overall negativity on the derivative, a problem
which is generally intractable). Still, we can obtain good approximations
by using the common approach of discretizing $k(t)$, a technique which is used
in numerous problems. Here, we discretize both the Fourier
transform $K(u)$ and the spatial coordinate $t$, as follows.

\begin{enumerate}[rightmargin=1cm]
\item Define a Fourier approximation to $k(t)$ as
  $\displaystyle \sum_{l=0}^n a_l \cos(l\pi t)$.
\item Impose the condition $k(0)=1$ by the linear
  equality $\displaystyle \sum_{l=0}^n a_l=1$.
  \item Impose the condition $k(\varepsilon)=c$ by the linear
  equality $\displaystyle \sum_{l=0}^n a_l \cos(l\pi \varepsilon)=c$.
\item Impose (approximate) monotonicity by demanding that the derivatives at a
  dense grid, e.g. \allowbreak $\{0,0.001 \ldots 0.999\}$, are all negative. Note
  that these are linear inequalities in the sought coefficients
  $a_l, l=0 \dots n$. Exact monotonicity can be directly checked on the resulting solution.
\item Impose positive definiteness of the kernel by the conditions
  $a_l \geq 0, l=0 \dots n$.
\item Impose positivity of the kernel by $k(1) \geq 0$, another
  linear constraint (this suffices, since $k(t)$ is monotonically
  decreasing).
\item under the constraints in 2-6 above, minimize $k(\delta)$,
  which is, also, a linear function of $a_l \geq 0, l=0 \dots n$.  
\end{enumerate}

Thus, the problem of finding an optimal kernel for our purpose has
been approximated by a linear programming problem. Increasing $n$ and
the grid over which the condition $k'(t) \leq 0$ is imposed, improves
the approximation.

Many experiments have indicated that in the "normal" range of parameters,
the optimal kernel is nearly equal to the simpler Gaussian kernel; this is next demonstrated.
\subsection{Examples of optimal kernels}
Below are two figures, comparing the basic GRF
to two optimal kernels, for given values of $\varepsilon, \delta, c$.
\medskip
\paragraph{Fig. \ref{fig:gaussian vs optimal}} 
The optimal and Gaussian kernels, for  $\varepsilon=0.01,c=0.99,\delta=0.1$.
    Note that the Gaussian kernel (blue) is very close to the optimal one (red), although (as can be expected) the optimal kernel achieves a slightly lower value at $\delta$.
    
\paragraph{Fig. \ref{fig:gaussian exp minus 5}}An example in which the optimal kernel is markedly better than the Gaussian:
    $\varepsilon=0.3,c=0.6,\delta=0.7$. Its "odd" shape notwithstanding, the red curve depicts a viable, doubly positive kernel. Interestingly, note that sharp "drop" in the kernel just before $\delta$, in which the optimal kernel  achieves a much lower value than the Gaussian one.
    
    However, the choice of parameters in Fig. \ref{fig:gaussian exp minus 5} is less appropriate for the practical decision problems we aim to solve in this paper; and typical values yield optimal kernels which are rather close to the basic  Gaussian kernel.

\begin{figure}
    \centering
    \begin{minipage}{0.45\textwidth}
        \centering
        \includegraphics[scale=0.4]{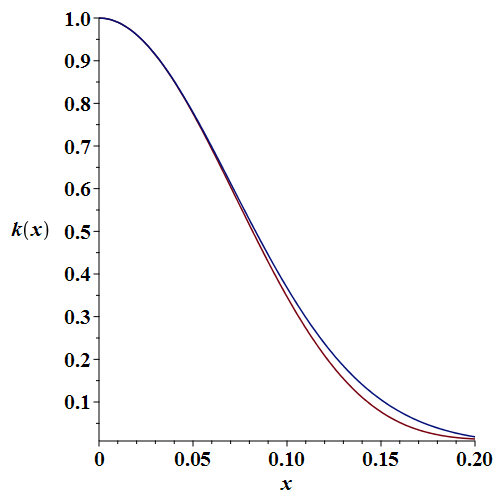}
        \caption{Comparison of the optimal (red) vs. Gaussian kernels (blue), for $\varepsilon=0.01,c=0.99,\lambda=0.1$. }
        \label{fig:gaussian vs optimal}
    \end{minipage}\hfill
    \begin{minipage}{0.45\textwidth}
        \centering
        \vspace{0.2cm}
        \includegraphics[scale=0.42]{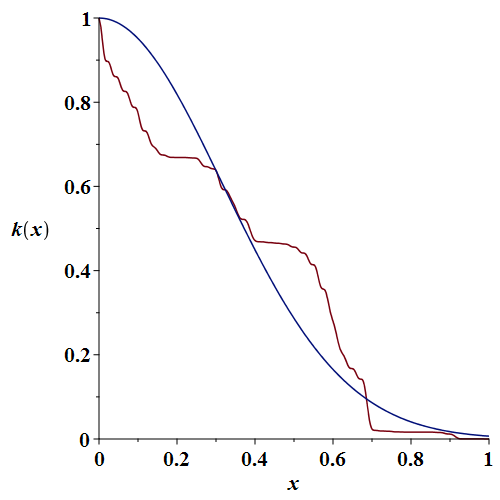}
        \caption{An example of an optimal kernel (red) which substantially differs from the Gaussian kernel (blue), for $\varepsilon=0.3,c=0.6,\lambda=0.7$.}
        \label{fig:gaussian exp minus 5}
    \end{minipage}
\end{figure}

\section{Optimal Kernels in Higher Dimensions} \label{section:higher_dimenstion}
In this section, we extend the approach of Section \ref{sec:bochner1d} to more than one variable, with the Fourier expansion being rvarepsilonlaced by its multivariate generalization, the Dini series. 
\subsection{Statement of the Problem}
Given a dimension $d \geq 1$ and parameters $c,\varepsilon \in [0,1]$, we denote by $\K_d(c,\varepsilon)$ the class of kernels $k:\R_+\to \R$ such that

\begin{enumerate}
\item $k(t)\geq 0$ for all $t\geq 0$;
\item $k$ is continuous and monotonically decreasing in the interval $[0,1]$;
\item The map $x\in \R^d \mapsto k(\|x\|)$ is positive semidefinite\footnote{We will often say that $k$ is positive semidefinite in $\r^d$.}, that is, for any collection of vectors $\{x_j\}_{j=1}^N\subset \R^d$ the matrix $[k(\| x_i - x_j\||)]_{i,j=1,...,N}$ is positive semidefinite;
\item $k(\varepsilon)=c$.
\end{enumerate}
Then, given $\delta\in (\varepsilon,1)$, we want to minimize the quantity $k(\delta)$ over $\K_d(c,\varepsilon)$. We let
$$
\kappa_d(c,\varepsilon,\delta) = \inf_{k\in \K_d(c,\varepsilon)} k(\delta).
$$

\begin{definition}
We say that $k_0\in \K_d(c,\varepsilon)$ is a global minima for $\delta$ if $ k(\delta)\geq k_0(\delta)$ for all $k \in \K_d(c,\varepsilon)$. In addition, we say $k$ is a unique if $ k(\delta)= k_0(\delta)$ implies $k=k_0$.
\end{definition}

Classical results of Bochner \cite{positive_positivedefinite} and Schoenberg \cite{Scho42} show that: 
\begin{enumerate}
\item These classes are nested: $\K_{d+1}(c,\varepsilon) \subset \K_{d}(c,\varepsilon)$.
\item $k:\r_+ \to \r$ is positive semidefinite and $k(0)=1$ if and only if
$$
k(\|x\|) = \int_{\r^d} e^{ix \cdot y} d\mu(y) 
$$
for all $x\in \r^d$, where $\mu$ is some probability measure in $\r^d$.
\item If we let $\K_{\infty}(c,\varepsilon):=\bigcap_{d\geq 1} \K_{d}(c,\varepsilon)$, then $k$ must be of the form
\begin{align}\label{eq:gauss-rvarepsilon}
k(\|x\|) = \int_0^\infty e^{-\la \|x\|^2} d\mu(\la) 
\end{align}
for some probability measure $\mu$ on $\r_+$.
\end{enumerate}

\begin{proposition}
Let $\kappa_\infty(c,\varepsilon,\delta) = \inf_{k\in \K_\infty(c,\varepsilon)} k(\delta)$. Then
$$
\kappa_\infty(c,\varepsilon,\delta)=c^{(\del/\varepsilon)^2}
$$
and the kernel $k(t)=e^{- t^2 \varepsilon^{-2}\log(1/c)}$ is the unique global minima.
\end{proposition}
\begin{proof}
By Schoenberg's results \eqref{eq:gauss-rvarepsilon} we have that 
$$
k(t) = \int_0^\infty e^{-\la t^2} d\mu(\la)
$$
for some probability measure $\mu$. In particular, since $\del/\varepsilon > 1$, Jensen's inequality gives 
$$
c^{(\del/\varepsilon)^2} = k(\varepsilon)^{(\del/\varepsilon)^2} = \left(\int_0^\infty e^{-\la \varepsilon^2} d\mu(\la) \right)^{(\del/\varepsilon)^2} \leq \int_0^\infty e^{-\la \del^2} d\mu(\la) = k(\del).
$$
Equality can only be attained above iff the function $\la \mapsto e^{-\la \varepsilon^2}$ is constant in the support of $\mu$, that is, $\mu$ is a Dirac delta at some point $\la_0$. Since $k(\varepsilon)=c$ we must have $\la_0=\varepsilon^{-2}\log(1/c)$. This completes the proposition.
\end{proof}

\subsection{Dini Series and the Parameter space for  \texorpdfstring{$\K_d(c,\varepsilon)$}{Kd(c,eps)}}\label{sec:diniseries}
It will be convenient to assume that our data is contained in a ball of radius $1/2$, hence we can restrict our attention to the values of $k$ for $t\in [0,1]$ only. Moreover, in working with the Dini series the way it is set up, once $k$ is determined via a \emph{finite} Dini series with nonnegative coefficients  in the interval $[0,1]$, we can easily extend $k$ to $\r_+$ using the series itself, since each function  $x \in \r^d \mapsto A_\nu(j_{\nu+1,n} \|x\|)$ is positive definite.

The Dini series\footnote{Sometimes also called Fourier-Bessel series.} of a function $f:[0,1]\to \r$  is given by
\begin{equation}\label{eq:Dini}
f(x)=\BB_0(x)+\sum_{n=1}^\infty c_n J_\nu(\lambda_n x),
\end{equation}
where $0<\lambda_1<\lambda_2<\ldots$ denote the positive zeros 
of the function

\smallskip

\begin{equation}\label{diniserisfunction}
zJ_\nu'(z)+HJ_\nu(z) = (H+\nu)J_\nu(z)-zJ_{\nu+1}(z).
\end{equation}
Here, $J_\nu$ is the Bessel function of the first kind of order $\nu\geq -\frac12$, and $H\in\r$ is some parameter. This Bessel function arises naturally as the Fourier transform of the unit sphere in $\r^d$. Indeed, letting 
$$
A_\nu(z):=\Gamma(\nu+1)(z/2)^{-\nu} J_\nu(z),
$$
we have
\begin{align}\label{id:Asphere}
A_{d/2-1}(\|x\|) = \frac{1}{\text{Area}(\sp^{d-1})}\int_{\sp^{d-1}} e^{i x\cdot y} dy
\end{align}
for all $x\in \r^d$. The Dini series which is most convenient for our purposes is the one with $H=-\nu $, so that equation \eqref{diniserisfunction} is simply $zJ_{\nu+1}(z)=0$ and $\{\la_n\}_{n\geq 1}$ are positive Bessel zeros of $J_{\nu+1}(z)$. Applying the identity \cite[5.11-(8)]{W66}
\begin{equation}\label{besselidentity}
\int_0^1 J_\nu(\al t)J_\nu(\be  t)t d t = \frac{\al J_{\nu+1}(\al)J_{\nu}(\be)-\be J_{\nu}(\al)J_{\nu+1}(\be)}{\al^2-\be^2}.
\end{equation}
for $\al=\la_m$ and $\be=\la_n$ are distinct zeros of \eqref{diniserisfunction}, then one sees that this integral vanishes whenever $m\neq n$. One can then rescale the results from \cite[ 18.33]{W66}, and invoke the identity \cite[5.11-(11)]{W66}, $\int_0^1 A_\nu^2(\lambda_n t) {t^{2\nu+1}} d t= \frac{A_\nu^2(\lambda_n)}2,$ in order to derive the following proposition.

\begin{proposition}\label{prop:diniseries}
Let $\nu\geq -\frac12$. For every measurable $f:[0,1]\to\cp$ with
$$
\|f\|_\nu:=\int_0^1 |f(t)|^2 {t^{2\nu+1}} d t <\infty
$$
 we have that
\begin{equation}\label{eq:DiniSeries}
f(x)=a_0 +  \sum_{n=1}^\infty a_n  {A_\nu(\lambda_n t )}
\end{equation}
with convergence in the norm $\|f\|_\nu$, where $\{\lambda_{n}\}_{n\geq 1}$ denote the positive zeros of the Bessel function $J_{\nu+1}$,
\begin{equation}\label{eq:DiniCoeff}
a_n= \frac{2}{A_\nu(\lambda_n) ^2}\int_0^1 f(t)A_\nu(\lambda_{n} t) {t^{2\nu+1}d t},
\end{equation}
for all $n\geq 1$, and
$a_0=\tfrac{1}{2(\nu+1)}\int_0^1 f(t) {t^{2\nu+1}d t}$.
Moreover, if $f$ is continuous and of bounded variation in $[0,1]$, then the Dini series \eqref{eq:DiniSeries} of $f$ converges absolutely and uniformly in $[0,1]$.
\end{proposition}

\begin{remark}For now on we set $\nu=d/2-1$, $\{j_{\nu,n}\}_{n\geq 1}=\text{Zeros}(J_{\nu}) \cap \r_+$ and $f_{d,n}(t)=A_\nu(j_{\nu+1,n}t)$ for $n\geq 1$ (note that $f_{1,n}(t)=\cos(\pi n t)$). Identity \eqref{id:Asphere} shows that any of the functions $x \in \r^d \mapsto f_{d,n}(\|x\|)$ is indeed positive definite.
\end{remark}

It is then straightforward to deduce the following result.

\begin{proposition}\label{prop:dinipositive}
Let $k:\r_+\to\r$ be a continuous, of bounded variation and positive semidefinite in $\r^d$. Then its Dini series 
$$
k(t)=a_0 + \sum_{n\geq 1} a_n f_{d,n}(t),
$$
for $t\in [0,1]$, has only nonnegative coefficients, that is, $a_n\geq 0$ for all $n\geq 0$.
\end{proposition}

Observing that for all $0<t<1$ we have $|f_{d,n}(t)|< f_{d,n}(0)=1$, we conclude that if $k\in \K_d(c,\varepsilon)$ then 
\begin{align}\label{formalternative}
k(t)=1 - \sum_{n\geq 1} a_n (1-f_{d,n}(t)).
\end{align}
Moreover, letting 
$
g_{d,n}(t)=\frac{1-f_{d,n}(t)}{1-f_{d,1}(t)}
$
we conclude that
\begin{align}\label{id:crucial}
k(t) = 1 - (1-c)\frac{1-f_{d,1}(t)}{1-f_{d,1}(\varepsilon)} + {(1-f_{d,1}(t))}\sum_{n\geq 2} a_n (g_{d,n}(\varepsilon)-g_{d,n}(t)).
\end{align}
An important thing to notice is that the function 
$$
\tilde k(t) = 1 - (1-c)\frac{1-f_{d,1}(t)}{1-f_{d,1}(\varepsilon)}
$$
belongs to $\K_d(c,\varepsilon)$ iff 
$$
c \geq \frac{f_{d,1}(\varepsilon)-f_{d,1}(1)}{1-f_{d,1}(1)}.
$$
To see this, it is easy to show using basic facts about Bessel functions that $f_{d,1}(t)$ is decreasing and that $f_{d,1}(1)<0$, and so $k_1$ is decreasing and the condition above is just $\tilde k(1)\geq 0$.

\section{Numerical Evaluation} \label{section:numerical_eval}
\subsection{Heuristic}\label{Heuristic}
 Identity \eqref{id:crucial} shows that if $g_{d,n}(\ep)> g_{d,n}(\delta) $ for all $n\geq 2$ then
$$
k(\delta) \geq \tilde k(\delta),
$$
and the kernel $\tilde k(t)$ is the unique global optimal. 
Numerically, one can see that the functions $g_{d,n}(t)$ decay fast when $n\geq 2$ and $n = O(d)$ (especially for large $d$).  
These functions are essentially decreasing for all $t\in [0,1]$ and this indicates there is no improvement in taking $a_n>0$ for $n=O(d)$. Moreover,  for $n \gg d$, the functions $g_{d,n}$ start to oscillate rapidly and it becomes very difficult to enforce the condition $k'(t)\leq 0$, so there is no advantage in setting $a_n>0$ for $n\gg d$ either. This explains what we observe in our numerical experiments of Section ~\ref{section:numerical_eval}: for small $\ep$ and large $\delta$, it is almost always the case that the kernel $\tilde k(t)$ is optimal, since the monotonicity of the $g_{d,n}$ sets $a_n \approx 0$ for $n=2,..,O(d)$ and the condition $k'(t)\leq 0$ forces the remaining $a_n$ to vanish.

Let $\E_d(c)$ be the inverse of the function $\ep \mapsto c=\frac{f_{d,1}(\ep)-f_{d,1}(1)}{1-f_{d,1}(1)}$. Note that $c\mapsto \E_d(c)$ is decreasing with $\E_d(0)=1$ and $\E_d(1)=0$. In particular, given $c$, the kernel 
$$
\tilde k(t) = 1 - (1-c)\frac{1-f_{d,1}(t)}{1-f_{d,1}(\ep)}
$$
belongs to $\K_d(c,\ep,\del)$ exactly when $\ep\geq \E_d(c)$. For most practical purposes we set $c\approx 0.9$. In the following table we estimate some values of $\E_d(c)$ for $c=0.9$

\hspace{-0.3cm}
\begin{center}
\begin{tabular}{ |c| c |c| c|c|c|c|c|c|c|c| }
 \hline $d$ &  1 & 2&3&4&5&6&7&8&9&10 \\  \hline
 $\E_d(0.9)$ &   0.2048 & 0.1991 & 0.1938 & 0.1890 & 0.1845 & 0.1803 & 0.1764 & 0.1728 & 0.1694 & 0.1661  \\  \hline
\end{tabular}
\end{center}
\medskip
It appears that the function $\E_d(c)$ decreases with $d$ for every fixed $c$.
This table shows that in low dimensions taking $\ep = 0.21$ for a $0.9\leq c\leq 1$ guarantees that $\tilde k(t) \in \K_d(c,\ep,\del) $. We have performed numerous numerical experiments for $d=1,...,10$, $\ep \gtrapprox \E_d(c)$, several different values of $c \lessapprox 1$ and $ \delta > \ep$, using the following setup.
\bigskip
\begin{mdframed}[nobreak=true,align=center]
{\bf Linear Programming Setup.} Given $N,M\geq 1$ 
\begin{align}\label{lp}
\begin{split}
\text{Minimize} \quad  \sum_{n= 0}^N a_n f_{d,n}(\delta)  \quad : \quad & a_n\geq 0 \ \ \text{for each } n=0,...,N; \quad  \sum_{n= 0}^N a_n=1; \\
& \sum_{n= 0}^N a_n f_{d,n}(\ep)=c; \quad 
 \sum_{n= 0}^N a_n f_{d,n}(1) \geq 0; \\
& \sum_{n= 1}^N a_n f_{d,n}'(m/M) \leq 0 \ \ \text{for each } m=0,...,M.
\end{split}
\end{align}
\end{mdframed}
Derivatives can be computed by the formula $f_{d,n}'(t) = -j_{\nu+1,n}^2 t A_{\nu+1}(j_{\nu+1,n} t) /d$. Good results only appear when we take $M> N^2$. In all cases, the kernel $\tilde k$ seem to be the optimal choice, confirming our heuristic assumption.
\begin{conjecture}[Heuristics]\label{conj}
If $c\lessapprox 1$ and $\ep\gtrapprox \E_d(c)$ then the kernel
$$
\tilde k(t) = 1 - (1-c)\frac{1-f_{d,1}(t)}{1-f_{d,1}(\ep)}
$$
is the unique global minima of $\K_d(c,\ep)$ for any $\del \in (\ep,1]$.
\end{conjecture}

The extreme case when $\ep \ll \E_d(c)$ seems to be quite challenging, even in the desired range of parameters such as $c\approx .9$ and $\del \approx 1$. It seems that one can actually make $k(t)=0$ for $t\geq \delta$ at the expense of $\ep$ being very small,  see Figure \ref{fig1}.

\begin{figure}
\centering
\begin{tikzpicture}[scale=.9]

\begin{axis}[
    at={(0.1\linewidth,0cm)}, xmin=0, xmax=1,
    ymin=0, ymax=1,
    xtick={0,0.1,0.3,0.5,0.7,0.9,1},
    ytick={.2,.4,.6,.8,1},
    xlabel={$(d,c,\ep,\del)=(1,0.99,0.01,0.10)$}
]
\addplot[color=blue,mark size=1.0pt] table {Examples/1_99_1_10};
\end{axis}

\begin{axis}[
    at={(0.61\linewidth,0cm)}, xmin=0, xmax=1,
    ymin=0, ymax=1,
    xtick={0,0.1,0.3,0.5,0.7,0.9,1},
    ytick={.2,.4,.6,.8,1},
    xlabel={$(d,c,\ep,\del)=(1,0.60,0.30,0.70)$}
]
\addplot[color=blue,mark size=1.0pt] table {Examples/1_60_30_70};
\end{axis}

\begin{axis}[
    at={(0.1\linewidth,7cm)}, xmin=0, xmax=1,
    ymin=0, ymax=1,
    xtick={0,0.1,0.3,0.5,0.7,0.9,1},
    ytick={.2,.4,.6,.8,1},
    xlabel={$(d,c,\ep,\del)=(5,0.90,0.10,0.70)$}
]
\addplot[color=blue,mark size=1.0pt] table {Examples/5_90_10_70};
\end{axis}

\begin{axis}[
    at={(0.61\linewidth,7cm)}, xmin=0, xmax=1,
    ymin=0, ymax=1,
    xtick={0,0.1,0.3,0.5,0.7,0.9,1},
    ytick={.2,.4,.6,.8,1},
    xlabel={$(d,c,\ep,\del)=(5,0.99,0.001,0.10)$}
]
\addplot[color=blue,mark size=1.0pt] table {Examples/5_99_1_10};
\end{axis}

\begin{axis}[
    at={(0.1\linewidth,14cm)}, xmin=0, xmax=1,
    ymin=0, ymax=1,
    xtick={0,0.1,0.3,0.5,0.7,0.9,1},
    ytick={.2,.4,.6,.8,1},
    xlabel={$(d,c,\ep,\del)=(10,0.90,0.10,0.70)$}
]
\addplot[color=blue,mark size=1.0pt] table {Examples/10_90_10_70};
\end{axis}

\begin{axis}[
    at={(0.61\linewidth,14cm)}, xmin=0, xmax=1,
    ymin=0, ymax=1,
    xtick={0,0.1,0.3,0.5,0.7,0.9,1},
    ytick={.2,.4,.6,.8,1},
    xlabel={$(d,c,\ep,\del)=(10,0.99,0.01,0.10)$}
]
\addplot[color=blue,mark size=1.0pt] table {Examples/10_99_1_10};
\end{axis}

\end{tikzpicture}
\caption{Examples of kernels for varying dimensions and parameters.}
\label{fig1}
\end{figure}
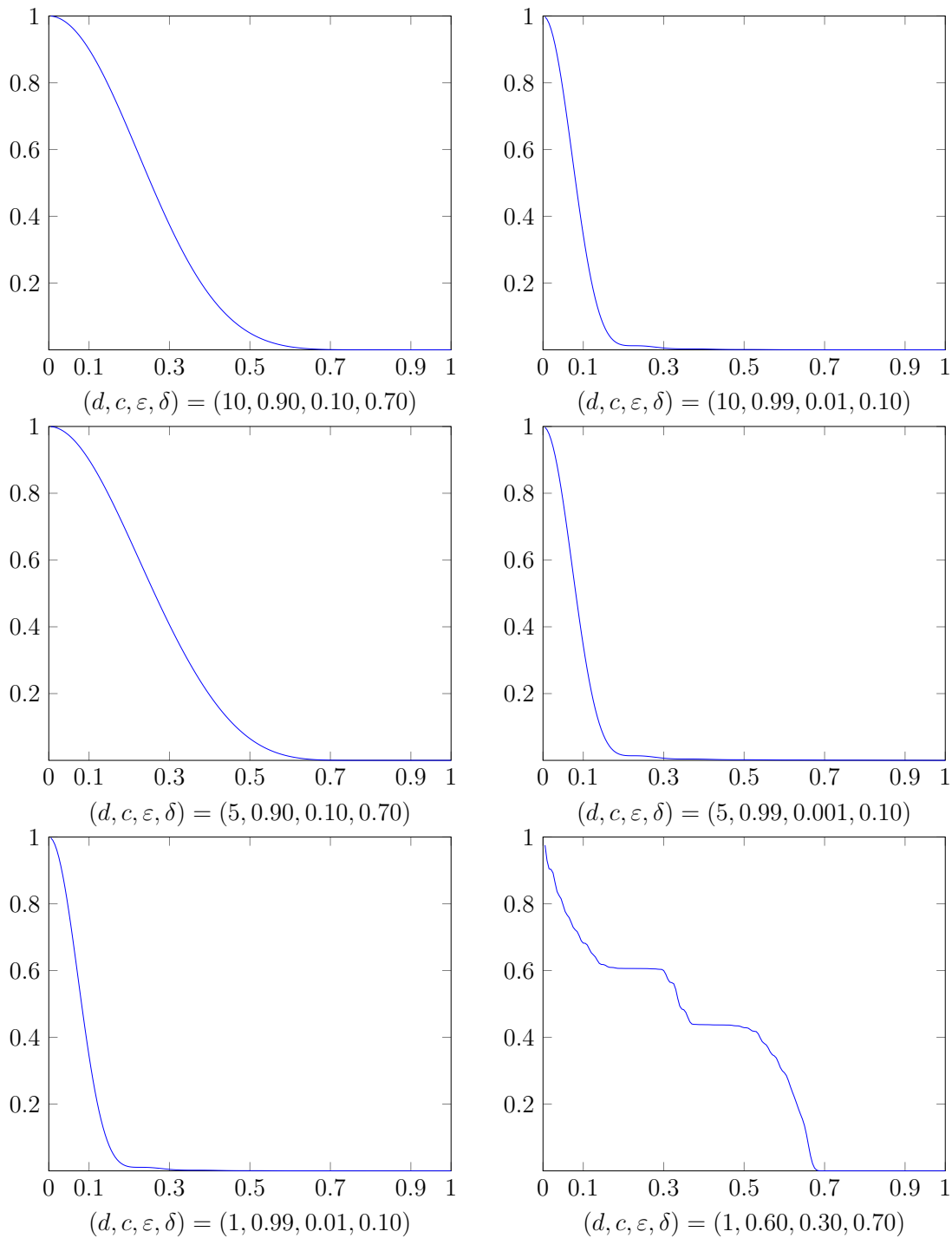

\section{Conclusions} \label{section:conclusions}

We have presented a probabilistic hashing method, based on random fields, which allows, for some parameter values, to differentiate between "clusterable" and "non-clusterable" data. The method easily extends to distributed, dynamic environments, requiring only a constant communication overhead, and maintaining a degree of privacy.

Future work will study further field types, as well as solving decision problems for other geometric properties (e.g intrinsic dimensionality). We also plan to study "softer" definitions of properties $M,C$ (e.g. which still hold if a certain percentage of the data is modified).

\bibliographystyle{abbrv}
\bibliography{main}

\end{document}